\documentclass[a4paper]{llncs}

\usepackage[utf8x]{inputenc}

\usepackage{logic9-thm}
\usepackage[textsize=small]{todonotes}
\usepackage{multicol}
\usepackage{multirow}
\usepackage{tikz}
\usetikzlibrary{arrows,automata}
\usetikzlibrary{positioning}
\usetikzlibrary{shapes.symbols}

\usepackage{listings}

\lstdefinelanguage{algo}{%
  morekeywords={function,push,pop,top,for,all,and,or,not,if,then,else,repeat,until,while,do,report,return,such,that,int,stack,end,delete,let,procedure,skip}
}

\newcommand{\downbox}[1]{\raisebox{-1ex}{\scriptsize $#1$}}
\newcommand{\norm}{\mathit{norm}}
\newcommand{\lleft}{\mathit{left}}
\newcommand{\rright}{\mathit{right}}
\newcommand{\tto}{\ensuremath{\Rightarrow}}
\newcommand{\abstr}{\ensuremath{\mathfrak{a}}}
\newcommand{\xra}[1]{\ensuremath{\xrightarrow{#1}}}

\newcommand{\extra}{\ensuremath{\mathsf{Extra}}}

\newcommand{\extraLUp}{\extra_{LU}^+}

\usepackage{enumerate}
\newcommand{\Rpos}{\mathbb{R}_{\ge 0}}
\newcommand{\val}{v}
\newcommand{\vali}{\mathbf{0}}
\newcommand{\CC}{\Phi}
\newcommand{\RR}{\mathbb{R}}
\renewcommand{\Acc}{F}
\renewcommand{\Nat}{\mathbb{N}}
\newcommand{\Integers}[0]{\ensuremath{\mathbb{Z}}}
\newcommand{\ZG}{\ensuremath{ZG}}

\newcommand{\abs}{\abstr}

\newcommand{\ggeq}{\succcurlyeq}
\newcommand{\lleq}{\preccurlyeq}

\title{Fast detection of cycles in timed automata}
\author{}
\institute{}
\begin{document}
\author{Aakash Deshpande\inst{1}, Fr\'ed\'eric Herbreteau\inst{2},  
  B. Srivathsan\inst{3}, Thanh-Tung Tran \inst{2} and Igor Walukiewicz\inst{2}}
\institute{
  Indian Institute of Technology Bombay, Mumbai, India
  \and
  Univ. Bordeaux, CNRS, LaBRI, UMR 5800, F-33400 Talence, France
  \and
  Chennai Mathematical Institute, Chennai, India
}

\maketitle

\begin{abstract}
  We propose a new efficient algorithm for detecting if a cycle in a
  timed automaton can be iterated infinitely often. Existing methods
  for this problem have a complexity which is exponential in the
  number of clocks. Our method is polynomial: it essentially does a
  logarithmic number of zone canonicalizations. This method can be
  incorporated in algorithms for verifying B\"uchi properties on timed
  automata. We report on some experiments that show a significant
  reduction in search space when our iteratability test is used.
\end{abstract}

\section{Introduction}


Timed automata~\cite{AD:TCS:1994} are one of the standard models of
timed systems. There has been an extensive body of work on the
verification of safety properties on timed automata. In contrast,
advances on verification of liveness properties are much less
spectacular due to fundamental challenges that need to be
addressed. For verification of liveness properties, expressed in a
logic like Linear Temporal Logic, it is best to consider a slightly
more general problem of verification of B\"uchi properties. This means
verifying if in a given timed automaton there is an infinite path
passing through an accepting state infinitely often.

Testing B\"uchi properties of timed systems can be surprisingly
useful.  We give an example in Section~\ref{sec:algorithm} where we
describe how with a simple liveness test one can discover a typo in
the benchmark CSMA-CD model. This typo removes practically all
interesting behaviours from the model.  Yet the CSMA-CD benchmark has
been extensively used for evaluating verification tools, and nothing
unusual has been observed. Therefore, even if one is interested solely
in verification of safety properties, it is important to ``test'' the
model under consideration, and for this B\"uchi properties are
indispensable.


Verification of timed automata is possible thanks to zones and their
abstractions
\cite{DT:TACAS:1998,Behrmann:STTT:2006,Herbreteau:CAV:2013}. Roughly,
the standard approach used nowadays for verification of safety
properties performs breadth first search (BFS) over the set of pairs
(state, zone) reachable in the automaton, storing only pairs with the
maximal abstracted zones (with respect to inclusion).

The fundamental problem in extending this approach to verification of
B\"uchi properties is that it is no longer sound to keep only maximal
zones with respect to inclusion. Laarman et
al.~\cite{Laarman:CAV:2013} recently studied in depth when it is sound
to use zone inclusion in nested depth first search (DFS). The
inability to make use of zone inclusions freely has a very important
impact on the search space: it can simply get orders of magnitude
larger, and this indeed happens on standard examples.

The second problem with verification of B\"uchi properties is that
fact that we need to use DFS instead of BFS to detect cycles. It has
been noted in numerous contexts that longer sequences of transitions
lead generally to smaller zones~\cite{Beh05}. This is why BFS would
often find the largest zones first, while DFS will get very deep into
the model and consider many small zones; these in turn will be made
irrelevant later with a discovery of a bigger zone closer to the
root. Once again, on standard examples, the differences in the size of
the search space between BFS and DFS exploration are often
significant~\cite{Beh05}, and it is rare to find instances where DFS
is better than BFS.


In this paper we propose an efficient test for checking
$\w$-iterability of a path in a timed automaton. By this we mean
checking if a given sequence of timed transitions can be iterated
infinitely often. We then use this test to ease the bottleneck created
by the two above mentioned problems. While we will still use DFS, we
will invoke $\w$-iterability test to stop exploration as early as
possible. In the result, we will not gain anything if there is no
accepting loop, but if there is one we will often discover it much
quicker.


\begin{figure}[tbhp]\label{fig:example}
  \centering
  \includegraphics[scale=.6]{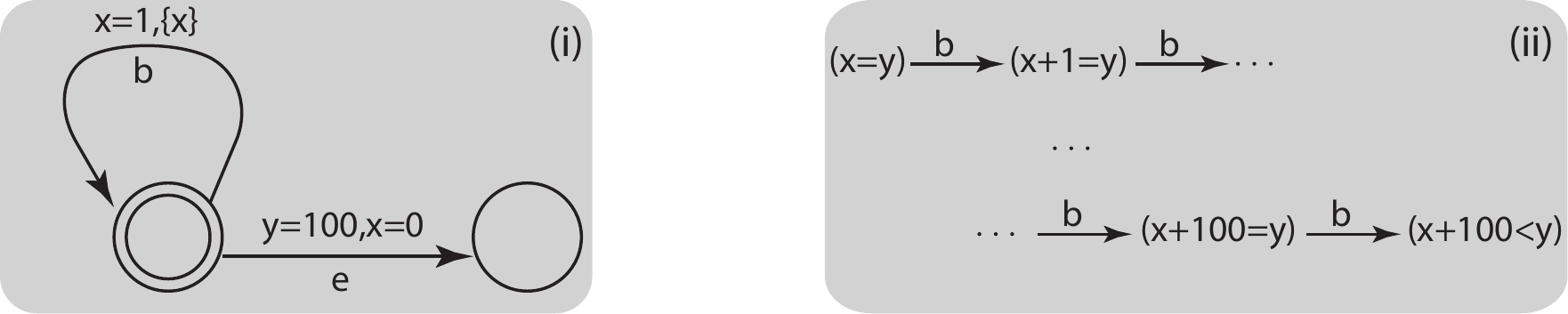}
  \caption{Iterability of $b$ transtion in timed automaton (i)
    requires a long exloration in the zone graph (ii).}
\end{figure}
An example of $\w$-iterability checking is presented in
Figure~\ref{fig:example}. The zone graph of the automaton has $100$
states due to the difference between $y$ and $x$ that gets bigger
after each iteration of $b$ transition. As the maximum constant for
$y$ is $100$, the zones obtained after $x + 100 = y$ will get
abstracted to $x + 100 < y$, giving a cycle on the zone graph. Without
iterability testing we need to take all these transitions to conclude
that $b$ can be iterated infinitely often. In contrast, our
iterability test applied to $b$ transition can tell us this
immediately.

A simple way to decide if a sequence of timed transitions $\s$ is
iterable is to keep computing successive regions along $\s$ till a
region repeats. However one might have to iterate $\s$ as many times
as the number of regions. Hence the bound on the number of iterations
given by this approach is $\Oo(|\s| \cdot M! \cdot 2^{n})$ where $n$
is the number of clocks and $M$ is the maximum constant occurring in
the automaton. Using zones instead of regions does not change much: we
still may need to iterate $\s$ an exponential number of times in $n$.

Our solution uses transformation matrices which are zone like
representations of the effect of a sequence of timed transitions. Such
matrices have already been used by Comon and Jurski in their work on
flattening of timed automata~\cite{ComJur99}. By analysing properties of
these matrices we show that $n^2$ iterations of $\s$ are sufficient
to determine $\w$-iterability. Moreover, instead of doing these
iterations one can simply do $\log(n^2)$ compositions of transformation
matrices. As a bonus we obtain a zone describing all the valuations
from which the given sequence of transitions is $\w$-iteratable.  The
complexity of the procedure is $\Oo((|\s| + \log n) \cdot n^3)$.

 One should bear in mind that $n$ is
usually quite small, as it refers only to active clocks in the
sequence. Recall for example that zone canonisation, that is
$\Oo(n^3)$ algorithm, is anyway evoked at each step of the exploration
algorithm. If we assume that arithmetic operations can be
done in unit time, the complexity of our algorithm does not depend on
$M$.

One should bear in mind that $n$ is usually quite small, as it refers
only to active clocks in the sequence. Recall for example that zone
canonicalisation, that is an $\Oo(n^3)$ algorithm, is anyway evoked at
each step of the exploration algorithm. If we assume that arithmetic
operations can be done in unit time, the complexity of our algorithm
does not depend on $M$.


Apart from the theoretical gains of the $\w$-iteratability test
mentioned above, we have also observed substantial gains on standard
benchmarks. We have done a detailed examination of standard timed
models as used, among others, in~\cite{Laarman:CAV:2013}. We explain
in Section~\ref{sec:algorithm} why when checking for false B\"uchi
properties that do not refer to time on these models, the authors of
op. cit. have observed almost immediate response. Our experiments show
that even on these particular models, if we consider a property that
refers to time, the situation changes completely. We present examples
of timed B\"uchi properties where iteration check significantly
reduces the search space.

\paragraph{Related work}

Acceleration of cycles in timed automata is technically the closest
work to ours. The binary reachability relation of a timed automaton
can be expressed in Presburger arithmetic by combining cycle
acceleration~\cite{Comon:CAV:1998,Bozga:ICALP:2006,Boigelot:CAV:2006}
and flattening of timed automata~\cite{ComJur99}. Concentrating on
$\w$-iterability allows to assume that all clock variables are reset
on the path. This has important consequences as variables that are not
reset can act as counters ranging from $0$ to the maximal constant $M$
appearing in the guards. Indeed, since checking emptiness of flat
automata is in \PTIME, general purpose acceleration techniques need to
introduce a blowup. By concentrating on a less general problem of
$\w$-iterability, we are able to find simpler and more efficient
algorithm.

The $\w$-iterability question is related to proving termination of
programs. A closely related paper is~\cite{BozIosKon12} where the
authors study conditional termination problem for a
transition given by a difference bound relation. The semantics of the
relation is different though as it is considered to be over integers
and not over positive reals as we do here. So, for example, a cycle
that strictly decreases the difference between two clocks is iterable
in our sense but not in~\cite{BozIosKon12}. More precisely,
Lemma~\ref{lem:w-iteratability} of this paper is not true over integers. The
decision procedure in op.\ cit.\ uses policy iteration algorithm, and
is exponential in the size of the matrix.

Tiwari~\cite{Tiw04} shows decidability of termination of simple loops
in programs
where variables range over reals and where the body is a set of linear
assignments. This setting is different from ours since programs are
deterministic and hence every state has a unique trajectory. Proving
termination and nontermination for a much more general class of
programs is addressed in~\cite{CooGulLev08,GupHenMaj08}

The already mentioned work of Laarman et al.~\cite{Laarman:CAV:2013}
examines in depth the problem of verification of B\"uchi properties of
timed systems. It focuses on parallel implementation of a modification
of the nested DFS algorithm. In Section~\ref{sec:algorithm} we report
on the experiments of using $\w$-iterability checking in the (single
processor) implementation of this algorithm.

\paragraph{Organization of the paper} Section~\ref{sec:preliminaries}
starts with required preliminaries. Section~\ref{sec:w-iteratability}
studies $\w$-iteratability, and describes the transformation graphs
that we use to detect $\w$-iterability. Certain patterns in these
graphs would help us conclude iteratability as shown in
Section~\ref{sec:pattern-makes-w}. Finally,
Section~\ref{sec:algorithm} gives some experimental results.


\section{Preliminaries}
\label{sec:preliminaries}

We adopt standard definition of timed automaton without diagonal
constraints. Below we state formally the B\"uchi non-emptiness
problem, and outline  how it is solved using zones and abstractions of
zones.  All the material in this section is well-known.

Let $\Rpos$ denote the set of non-negative reals. A \emph{clock} is a
variable that ranges over $\Rpos$. Let $X =\{x_0, x_1, \dots, x_n \}$
be a set of clocks. The clock $x_0$ will be special and will represent
the reference point to other clocks. A \emph{valuation} is a function
$\val\,:\,X\rightarrow\Rpos$ that is required to map $x_0$ to $0$. The
set of all clock valuations is denoted by $\Rpos^X$. We denote by
$\vali$ the valuation that associates $0$ to every clock in $X$.

A \emph{clock constraint} $\phi$ is a conjunction of constraints of
the form $x \sim c$ where $x\in X \setminus \{x_0\}$, $\sim
\in\{<,\leq,=,\geq,>\}$ and $c\in \Nat$. For example, $(x_1 \le 3
\wedge x_2 > 0)$ is a clock constraint. Let $\CC(X)$ denote the set of
clock constraints over the set of clocks $X$.  We write $\val\sat\phi$
when the valuation satisfies the constraint. We denote by $\val+\d$ the
valuation where $\d$ is added to every clock; and by $[R]\val$ we denote
the valuation where all clocks from $R$ are set to $0$.

  A \emph{Timed B\"uchi Automaton (TBA in short)}~\cite{AD:TCS:1994}
  is a tuple $\Aa=(Q,q_0,X,T,\Acc)$ in which $Q$ is a finite set of
  states, $q_0$ is the initial state, $X$ is a finite set of clocks,
  $\Acc \subseteq Q$ is a set of accepting states, and $T\,\subseteq\,
  Q\times\CC(X)\times 2^X \times Q$ is a finite set of transitions of
  the form $(q,g,R,q')$ where $g$ is a clock constraint called the
  \emph{guard}, and $R$ is a set of clocks that are \emph{reset} on
  the transition from $q$ to $q'$.

The semantics of a TBA $\Aa = (Q, q_0, X, T, \Acc)$ is given by a
transition system of its configurations. A \emph{configuration} of
$\Aa$ is a pair $(q,v) \in Q \times \Rpos^X$, with $(q_0, \vali)$
being the initial configuration. The transitions here are of the form:
$(q,v)\xra{t}^{\downbox{\d}}(q',v')$ for some transition $t =
  (q,g,R,q')\in T$ such that $(v+\d)\sat g$ and $v'=[R](v+\d)$.

A \emph{run} of $\Aa$ is a (finite or infinite) sequence of
transitions starting from the initial configuration $(q_0,
\vali)$.  
A configuration $(q, v)$ is said to be \emph{accepting} if $q \in
\Acc$. An infinite run \emph{satisfies the B\"uchi condition} if it
visits accepting configurations infinitely often. The run is
\emph{Zeno} if time does not diverge, that is, $\sum_{i\geq 0} \d_i
\le c$ for some $c \in \Rpos$. Otherwise it is \emph{non-Zeno}. The
problem we are interested is termed the \emph{B\"uchi non-emptiness
  problem}.


\begin{definition}
  \label{defn:language_tba}
  The \emph{B\"uchi non-emptiness problem} for TBA is to decide if
  $\Aa$ has a non-Zeno run satisfying the B\"uchi condition.
\end{definition}

The B\"uchi non-emptiness problem is known to be
\PSPACE-complete~\cite{AD:TCS:1994}. All solutions to this problem
construct an untimed B\"uchi automaton with an equivalent
non-emptiness problem. We describe such a translation below.

The standard algorithms on timed automata consider special sets of
valuations called \emph{zones}. A zone is a set of valuations
described by a conjunction of two kinds of constraints: either $x_i
\sim c$ or $x_i - x_j \sim c$ where $x_i, x_j \in X$, $c \in
\Integers$ and $\sim \in \{ <, \le, =, >, \ge \}$. For example $x_1 >
3$ and $x_2 - x_1 \le -4$ is a zone. Zones can be efficiently
represented by Difference Bound Matrices
(DBMs)~\cite{Dill:AVMFSS:1989}.

The \emph{zone graph} $\ZG(\Aa)$ of a TBA $\Aa$ is a B\"uchi automaton
$(S, s_0, \tto, F)$, where $S$ is the set of states, $s_0$ is the
initial state and $\tto$ is the transition relation. Each state in $S$
is a pair $(q, Z)$ consisting of a state $q$ of the TBA and a zone
$Z$. The initial node $s_0$ is $(q_0, Z_0)$ where $Z_0 = \{ \vali+
\d~|~ \d \in \Rpos\}$. For every $t = (q, g, R, q') \in T$, there
exists a transition $\tto^t$ from a node $(q,Z)$ as follows:
\begin{align*}
  (q, Z) \tto^t (q', Z') \qquad \text{where } Z' =\set{ \val'~|~
    \exists \val \in Z, ~\exists \d \in \Rpos: (q,\val) \xra{t}^{\d}
    (q',\val')}
\end{align*}
The transition relation $\tto$ is the union of $\tto^t$ over all $t
\in T$. It can be shown that if $Z$ is a zone, then $Z'$ is a
zone. 
Although the zone graph $\ZG(\Aa)$ deals with sets of valuations
instead of valuations themselves, the number of zones is still
infinite~\cite{DT:TACAS:1998}.

For effectiveness, zones are further abstracted. An \emph{abstraction
  operator} is a function $\abstr:\Pp(\Rpos^{|X|})\to\Pp(\Rpos^{|X|})$
such that $W\incl\abstr(W)$ and $\abstr(\abstr(W))=\abstr(W)$ for
every set of valuations $W \in \Pp(\Rpos^{|X|})$.  If $\abstr$ has a
finite range then this abstraction is said to be finite.  An
abstraction operator defines an abstract symbolic semantics:
\begin{align*}
  (q,Z) \tto^t_{\abs }(q',\abs(Z')) \qquad \text{ when $\abs(Z)=Z$ and
    $(q,Z)\tto^t (q',Z')$ in $\ZG(\Aa)$}
\end{align*}
We define a transition relation $\tto_\abs$ to be the union of
$\tto^t_\abs$ over all transitions $t$.  Given a finite abstraction
operator $\abs$, the \emph{abstract zone graph} $\ZG^\abs(\Aa)$ is a
B\"uchi automaton whose nodes consist of pairs $(q, Z)$ such that $Z =
\abs(Z)$. The initial state is given by $(q_0, \abs(Z_0))$ where
$(q_0, Z_0)$ is the initial state of $\ZG(\Aa)$. Transitions are given
by the $\tto_\abs$ relation. The accepting states are nodes $(q, Z)$
with $q \in F$. 

The abstraction operator $\extraLUp$~\cite{Behrmann:STTT:2006} is a
commonly used efficient abstraction operator. For concreteness we will
adopt this operator in the paper, but our approach can be also used
with other abstraction operators as for example
$a_{LU}$~\cite{Herbreteau:LICS:2012}. Thus we will consider the Buchi automaton given
by $\ZG^{\extraLUp}(\Aa)$, which is the abstract zone graph of $\Aa$ with respect to
$\extraLUp$. The following fact guarantees soundness of this approach.




\begin{theorem}[\cite{Li:FORMATS:2009}]\label{thm:LU-for-Buchi}
  A timed B\"uchi automaton $\Aa$ has a B\"uchi run iff the finite
  B\"uchi automaton $\ZG^{\extraLUp}(\Aa)$ has one.
\end{theorem}


Theorem~\ref{thm:LU-for-Buchi} gives an algorithm for non-emptiness of
TBA: given a timed B\"uchi automaton $\Aa$, it computes the (finite) B\"uchi
automaton $\ZG^{\extraLUp}(\Aa)$ and checks for its
emptiness. The main problem with this approach is efficiency. 
When checking for reachability it enough to consider
only maximal zones with respect to inclusion of
$\ZG^{\extraLUp}(\Aa)$. This optimization gives very important
perfomance gains, but unfortunately it is not sound for verification
of B\"uchi properties~\cite{Laarman:CAV:2013}. The above theorem does
not address the non-Zenoness issue. Since this issue does not influence our
$\w$ iterability test, we defer it to the conclusions. 



 
\section{$\w$-iterability}
\label{sec:w-iteratability}
Since we cannot use zone inclusion to cut down the search space, it
can very well happen that an exploration algorithm comes over a path
that can be iterated infinitely often but it is not able to detect it
since the initial and final zones on the path are different. In this
section we show how to test, in a relatively efficient way, if a
sequence of timed transitions can be iterated infinitely often
(Theorem~\ref{thm:iteratability-test}).

Consider a sequence of transitions $\s$ of the form
$\xra{t_1}\dots\xra{t_k}$, and suppose that $(q,Z)\tto^\s (q,Z')$. If
$Z\incl Z'$ then after executing $\s$ from $(q,Z')$ we obtain
$(q,Z'')$ with $Z''$ not smaller than $Z'$. So we can execute $\s$ one
more time etc. The challenging case is when $Z\not\incl Z'$.  The
procedure we propose will not only give a yes/no answer but will
actually compute the zone representation of the set of valuations from
which the sequence can be iterated infinitely often.  We will start by
making the notion of $\w$-iterability precise. 

An \emph{execution} of a sequence $\s$ of the form
$\xra{t_1}\dots\xra{t_k}$ is a sequence of valuations $v_0,\dots,v_k$
such that for some $\d_1,\dots\d_k\in \RR^+$ we have
\begin{equation*}
  v_0\xra{t_1}^{\d_1 }v_1\xra{t_2}^{\d_2 }\dots\xra{t_{k}}^{\d_{k} }v_k
\end{equation*}
In this case we write $v_0\xra{\s}^\d v_k$ where $\d=\d_1+\dots+\d_k$.
The sequence $\s$ is \emph{executable from $v$} if there is a sequence
of valuations $v_0, \dots, v_k$ as above with $v=v_0$. (For clarity of
presentation, we choose not to write the state component of
configurations $(q,v)$ and instead write $v$ alone.)

\begin{definition}
  The sequence $\s$ is \emph{$\w$-iterable} if there is an infinite
  sequence of valuations $v=v_0,v_1,\dots$ and an infinite sequence of
  delays $\d_1,\d_2,\dots$ such that
  \begin{equation*}
    v_0\xra{\s}^{\d_1 }v_1\xra{\s}^{\d_2 }\dots
  \end{equation*}
  We also say that the sequence $\s$ is $\w$-iterable from $v_0$.
\end{definition}

Using region abstraction one can observe the following
characterization of $\w$-iteratability in terms of finite executions.
Observe that the lemma talks about $\w$-iterability not about
iterability from a particular valuation (that would give a weaker
statement).
\begin{lemma}
  \label{lem:w-iteratability}
  A sequence of transitions is $\w$-iterable iff for every
  $n=1,2,\dots$ the $n$-fold concatenation of $\s$ has an execution.
\end{lemma}
\begin{proof}
  Suppose a sequence of transitions $\sigma$ is $\w$-iterable. Then,
  clearly every $n$-fold concatenation has an execution.

  Suppose every $n$ fold concatenation of $\s$ has an execution. Let
  $r$ be the number of regions with respect to the biggest constant
  appearing in the automaton. By assumption, the $r$-fold
  concatenation of $\s$ has an execution:
  \begin{align*}
    v_0 \xra{\s}^{\d_1} v_1 \xra{\s}^{\d_2} v_2 \dots v_{r-1}
    \xra{\s}^{\d_{r-1}} v_{r}
  \end{align*}
  In the above execution, there exist $i,j$ such that $v_i$ and $v_j$
  belong to the same region $R$, thus giving a cycle in the region
  graph of the form $R \xra{(j-i)\s} R$. By standard properties of
  regions, $\s$ is $\w$-iterable from every valuation in $R$.\qed
\end{proof}

The above proof shows that $\s$ is $\w$-iterable iff it is
$(r+1)$-iterable, where $r$ is the number of regions. As $r$ is
usually very big, testing $(r+1)$-iterablity directly is not a
feasible solution. Later we will show a much better bound than
$(r+1)$, and also propose a method to avoid checking $k$-iterability
by directly executing all the $\s^k$ transitions.

Our iterability test will first consider some simple cases when the
answer is immediate using which it will remove clocks that are not
relevant. This preprocessing step is explained in the following two
facts. We use $X$ for the set of clocks appearing on a sequence of
transitions $\s$, and $X_0$ for the set of clocks that are reset on
$\s$. Let $\fleq$ denote either $<$ or $\le$, and let $\ggeq$ stand
for $>$ or $\ge$.
\begin{lemma}
  \label{lem:conditions-non-iteratable}
  A sequence of transitions $\s$ is not $\w$-iterable if it satisfies
  one of the following conditions:
  \begin{itemize}
  \item for some clock $y \in X \setminus X_0$, we have guards $y
    \fleq d$ and $y \ggeq c$ in $\s$ such that either $d < c$, or $(d
    = c \text{ and } \fleq = <)$, or $(d = c \text{ and } \ggeq = >)$,
  \item there exist clocks $x \in X_0$ and $y \in X \setminus X_0$
    involved in guards of the form $y \lleq d$ and $x \ggeq c$ such
    that $c > 0$,
  \item for some clock $x \in X_0$ there is a guard $x > 0$ in $\s$
    and for some clock $y \in X \setminus X_0$, we have both the
    guards $y \le c$ and $y \ge c$ in $\s$.
  \end{itemize}
\end{lemma}
\begin{proof}
  If the first condition is true, it is clear that $\s$ cannot be
  iterated even twice. In the second case, at least $c > 0$ time units
  needs to be spent in each iteration. Since $y$ is not reset, the
  value of $y$ would become bigger than $d$ after a certain number of
  iterations and hence the guard $y \lleq d$ will not be satisfied. In
  the third situation, the guard $x > 0$ requires a compulsory
  non-zero time elapse. However as $d = c$ and $y$ is not reset, a
  time-elapse is not possible and thus $\s$ is not iterable.
\end{proof}

\begin{lemma}
  \label{lem:no-time-elapse-iteratable}
  Suppose $\s$ does not satisfy the conditions given in
  Lemma~\ref{lem:conditions-non-iteratable}. If, additionally, $\s$
  has no guard of the form $x \ggeq c$ with $c > 0$ for clocks $x$
  that are reset, then $\s$ is iterable.
\end{lemma}
\begin{proof}
  Suppose no guard of the form $x \ggeq c$ is present for clocks $x
  \in X_0$. As $\s$ does not satisfy the conditions in
  Lemma~\ref{lem:conditions-non-iteratable}, we can find for each $y
  \in X \setminus X_0$, a value $\l_y$ that satisfies all guards
  involving $y$ in $\s$. Set $v(x) = 0$ for $x \in X_0$ and $v(y) =
  \l_y$ for all $y \in X \setminus X_0$. Consider an execution without
  time elapse from $v$. Constraints $z \lleq k$ for clocks that are
  not reset are satisfied as we started with $v(z) \lleq k$, and
  $v(z)$ has not changed.  Constraints $z \lleq k$ for clocks that are
  reset also hold as these clocks take the constant value zero.  We
  get back $v$ after the execution implying that $\s$ is iterable from
  $v$.

  Suppose no guard of the form $x \ggeq c_x$ with $c_x > 0$ is present
  for the reset clocks $x \in X_0$, but for some clock $x \in X_0$ the
  guard $x>0$ appears in $\s$. This says that a compulsory time elapse
  is required.  Again, as Conditions 1 and 3 of
  Lemma~\ref{lem:conditions-non-iteratable} are not met by $\s$, we
  can pick a $\l_y$ that satisfies all guards in $\s$ that involve
  $y$. Set $v(x) = 0$ for $x \in X_0$ and $v(y) = \l_y$ for all $y \in
  X \setminus X_0$. From among all clocks $z \in X$, we find the
  minimum constant out of $k-v(z)$ where $z$ has a guard of the form
  $z \lleq k$.  In our execution we elapse time that is half of this
  minimum, after the reset.  Constraints $z>c$ for clocks that are not
  reset are satisfied.  Constraints $z<k$ for all clocks are satisfied
  by the way in which we choose our time delay.  We obtain a similar
  valuation and can repeat this procedure to execute $\s$ any number
  of times. We can conclude by Lemma~\label{lem:w-iteratability}.\qed
\end{proof}

\begin{proposition}\label{prop:pre-process}
  Let $\s$ be a sequence of transitions. Then:
  \begin{itemize}
  \item if $\s$ satisfies some condition in
    Lemma~\ref{lem:conditions-non-iteratable}, then $\s$ is not
    $\w$-iterable;
  \item if $\s$ does not satisfy
    Lemma~\ref{lem:conditions-non-iteratable} and does not contain a
    guard of the form $x \ggeq c$ with $c >0$ for a clock that is
    reset, then $\s$ is $\w$-iterable;
  \item if $\s$ does not satisfy
    Lemma~\ref{lem:conditions-non-iteratable} and contains a guard of
    the form $x \ggeq c$ with $c > 0$ for some clock $x$ that is
    reset, then $\s$ is $\w$-iterable iff $\s_{X_0}$ obtained from
    $\s$ by removing all guards on clocks not in $X_0$ is
    $\w$-iterable.
  \end{itemize}
\end{proposition}
\begin{proof}
  The first two cases have been shown in
  Lemmas~\ref{lem:conditions-non-iteratable} and
  \ref{lem:no-time-elapse-iteratable}. For the third case, note that
  the transition sequence $\s_{X_0}$ is the same as $\s$ except for
  certain guards that have been removed. Hence if $\s$ is iterable,
  the same execution is also an execution of $\s_{X_0}$. We will hence
  concentrate on the reverse direction.

  Suppose there exists a guard $x \ggeq c_x$ with $c_x > 0$ for a
  clock $x \in X_0$. From Condition 2 of Lemma
  \ref{lem:conditions-non-iteratable}, all guards involving $y \in X
  \setminus X_0$ would be of the form $y \ggeq k$. In other words, the
  upper bound guards involve only the clocks that are reset. Now,
  suppose $\s_{X_0}$ is $\w$-iterable from a valuation $v$. Let $\rho$
  be this execution.  We will now extend $v$ to a valuation $v'$ over
  the clocks $X$. For all clocks $x \in X_0$, let $v'(x) = v(x)$. For
  clocks $y \in X \setminus X_0$, set $v'(y)$ to a value greater than
  maximum of $k$ from among constraints $y \ggeq k$. By elapsing the
  same amounts of time as in $\rho$, the sequence $\s$ can be executed
  from $v'$. From $\rho$, we know that all clocks in $X_0$ satisfy the
  guards. All clocks $y \in X \setminus X_0$ satisfy guards $y \ggeq
  c_y$ by construction. As there are no other type of guards for $y$,
  we can conclude that $\s$ is iterable.

\end{proof}

We can thus eliminate all clocks that are not reset while discussing
iterability.  Our $\s$ now only contains clocks that are reset at
least once, and has some compulsory minimum time elapse in each
iteration.

 
\subsection{Transformation graphs}
\label{sec:transf-graphs}

Recall that the goal is to check if a sequence $\s$ of transitions $\xra{t_1}
\dots \xra{t_k}$ is $\w$-iteratable. From the results of the previous
section, we can assume that every clock is involved in a guard and is
also reset in $\s$. In the normal forward analysis
procedure, one would start with some zone $Z_0$ and keep computing the
zones obtained after each transition: $Z_0 \xra{t_1} \dots \xra{t_k}
Z_k$. The zone $Z_0$ is a set of difference constraints between
clocks. After each transition, a new set of constraints is
computed. This new constraint set reflects the changes in the clock
differences that have happened during the transition. Our aim is to
investigate if certain ``patterns'' in these changes are responsible
for non-iteratability of the sequence. To this regard, we associate to
every transition sequence what we call a \emph{transformation graph},
where the changes happening during the transitions are made more
explicit. In this subsection we will introduce transformation graphs
and explain how to compose them.

We start with a preliminary definition. Fix a set of clocks
$X=\set{x_0,x_1,\dots,x_n}$. As mentioned before, the clock $x_0$ will
be special as it will represent the reference point to other
clocks. We will work with two types of valuations of clocks. The
standard ones are $v:X\to \RR^+$ and are required to assign $0$ to
$x_0$. A \emph{loose} valuation $v$ is a function $v:X\to \RR$ with
the requirement that $v(x_i)\geq v(x_0)$ for all $i$.  In particular,
by subtracting $v(x_0)$ from every value we can convert a loose
valuation to a standard one. Let $\norm(v)$ denote this standard
valuation.

\begin{definition}
  A \emph{transformation graph} is a weighted directed graph whose vertices are
  $V=\set{0,\dots,k}\times X$ for some $k$, and whose edges have
  weights of the form $(\leq,d)$ or $(<,d)$ for some
  $d\in\Nat$.  
\end{definition}
We will say that a vertex $(0,x)$ is in the leftmost column of a
transformation graph and a vertex $(k,x)$ is in the rightmost
column. The graph as above has $k+1$
columns. Figure~\ref{fig:trans-graph} shows an example of a
transformation graph with three columns.

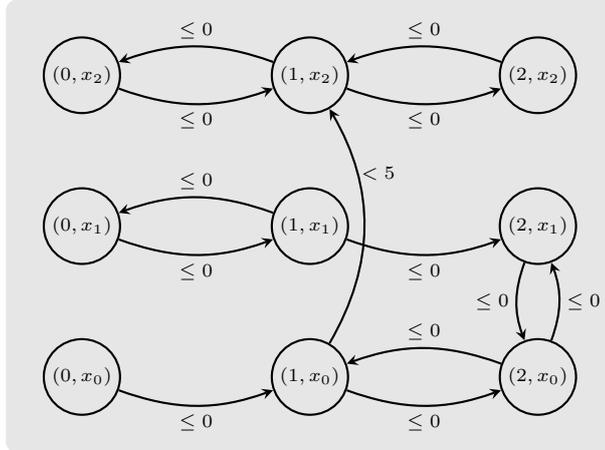
\begin{figure}[!h]
\centering
\begin{tikzpicture}

\fill[gray!20, rounded corners] (-1,-1) rectangle (7,5);

\begin{scope}[shorten >=1pt,node distance=2.5,on
      grid,auto, every node/.style = {shape=circle}, thick, inner
      sep=2pt]
\node [draw] (00) at (0,0) {\scriptsize $(0,x_0)$};
\node [draw] (01) at (0,2) {\scriptsize $(0,x_1)$};
\node [draw] (02) at (0,4) {\scriptsize $(0,x_2)$};
\node [draw] (10) at (3,0) {\scriptsize $(1,x_0)$};
\node [draw] (11) at (3,2) {\scriptsize $(1,x_1)$};
\node [draw] (12) at (3,4) {\scriptsize $(1,x_2)$};
\node [draw] (20) at (6,0) {\scriptsize $(2,x_0)$};
\node [draw] (21) at (6,2) {\scriptsize $(2,x_1)$};
\node [draw] (22) at (6,4) {\scriptsize $(2,x_2)$};

\end{scope}

\begin{scope}[->, >=stealth, thick]
\draw (02) edge [bend right = 20] (12);
\draw (12) edge [bend right = 20] (02);
\draw (01) edge [bend right = 20] (11);
\draw (11) edge [bend right = 20] (01);
\draw (00) edge [bend right = 20] (10);
\draw (12) edge [bend right = 20] (22);
\draw (22) edge [bend right = 20] (12);
\draw (11) edge [bend right = 20] (21);
\draw (10) edge [bend right = 20] (20);
\draw (20) edge [bend right = 20] (10);
\draw (10) edge [bend right] (12);
\draw (20) edge [bend right =20] (21);
\draw (21) edge [bend right =20] (20);
\end{scope}

\begin{scope}
\node at (1.5, 4.6) {\scriptsize $\le 0$};
\node at (1.5, 3.4) {\scriptsize $\le 0$};
\node at (4.5, 4.6) {\scriptsize $\le 0$};
\node at (4.5, 3.4) {\scriptsize $\le 0$};
\node at (1.5, 2.6) {\scriptsize $\le 0$};
\node at (1.5, 1.4) {\scriptsize $\le 0$};
\node at (4.5, 1.4) {\scriptsize $\le 0$};
\node at (1.5, -0.6) {\scriptsize $\le 0$};
\node at (4.5, 0.6) {\scriptsize $\le 0$};
\node at (4.5, -0.6) {\scriptsize $\le 0$};
\node at (3.9, 2.7) {\scriptsize $< 5$};
\node at (5.4, 1) {\scriptsize $\le 0$};
\node at (6.6, 1) {\scriptsize $\le 0$};
\end{scope}

\end{tikzpicture}
\caption{A transformation graph over the set of clocks $\{x_0, x_1,
  x_2\}$. }
\label{fig:trans-graph}
\end{figure}


Edges in a transformation graph represent difference constraints. For
instance, in the graph from Figure~\ref{fig:trans-graph}, the edge
from $(1, x_0)$ to $(1,x_2)$ with weight $< 5$ represents the
constraint $(1, x_2) - (1, x_0) < 5$. We will now formally define what
a solution to a transformation graph is.

\begin{definition}
  Let $G$ be a transformation graph with $k+1$ columns. A
  \emph{solution} for $G$ is a sequence of loose valuations
  $v_0,\dots,v_k: X\to\RR$ such that for every edge from $(i,x_j)$ to
  $(p,x_q)$ of weight $\fleq d$ in $G$ we have $v_p(x_q
  )-v_i(x_j)\fleq d$, where $\fleq$ stands for either $\le$ or $<$.
\end{definition}
A solution to the transformation graph in Figure~\ref{fig:trans-graph}
would be the sequence $v_0, v_1, v_2$ where $v_0: = \langle -1.5, 5, 1
\rangle$, $v_1: = \langle -3, 5, 1 \rangle$, $v_2 := \langle -3, -3,
3.5 \rangle$ (we write the values in the order $\langle x_0, x_1, x_2
\rangle$). Check for instance that $v_1(x_2) - v_1(x_0) < 5$ according
to the definition. 

The following lemma describes when a transformation graph has a
solution (cf. See proof of Proposition 1 in \cite{NonConvex-ArXiv}).

\begin{lemma}
  \label{lem:negative-cycles}
  A transformation graph $G$ has a solution iff it does not have a
  cycle of a negative weight.
\end{lemma}

Our aim is to construct transformation graphs that reflect the changes
happening during a transition sequence. Following is the definition of
what it means to reflect a transition sequence. 

\begin{definition}
\label{def:reflecting-a-transition}
  We say that a transformation graph $G$ reflects a transition
  sequence $\s$ if the following hold:
  \begin{itemize}
  \item For every solution $v_0,\dots,v_k$ of $G$ we have
    $\norm(v_0)\xra{\s}^\d\norm(v_k)$ where $\d=-v_k(x_0)-v_0(x_0)$.
  \item For every $v_0 \xra{\s}^\d v_k$ the pair of valuations $v_0$,
    and $(v_k-\d)$ can be extended to a solution of $G$.
  \end{itemize}
\end{definition}
Consider Figure~\ref{fig:trans-graph} again. It reflects the single
transition $t:~\xra{~x_2 < 5,~\{ x_1 \}~}$ with the guard $x_2 < 5$ and
the reset $\{x_1\}$. Let us illustrate our claim with an example. For the solution $v_0, v_1,
v_2$ given above, check that $\norm(v_0) \xra{t}^{1.5}
\norm(v_2)$. Similarly, for every pair of valuations that satisfy $v
\xra{t}^\d v'$, one can check that $v$ and $v'-\d$ can be extended to
a solution of the graph in Figure~\ref{fig:trans-graph}.

We will now give the construction of the transformation graph for an
arbitrary transition. Subsequently, we will define a composition
operator that will extend the construction to a sequence of transitions.

\begin{definition}
\label{def:trans-graph-of-a-transition}
Let $t:~\xra{~g~,R~}$ be a transition with guard $g$ and reset
$R$. The transformation graph $G_t$ for $t$ is a 3 column graph with
vertices $\{0,1,2\} \times X$. Edges are defined as below.
\paragraph{Time-elapse + guard edges:}
\begin{enumerate}
\item $(0, x_0) \xra{~\le 0~} (1, x_0)$,
\item $(0, x_i) \xra{~\le 0~} (1,x_i)~$ and $~(1, x_i) \xra{~\le 0~} (0,
  x_i)$ for all $x_i \neq x_0$,
\item $(1, x_0) \xra{~\fleq c~} (1, x_i)$ for every constraint $x_i
  \fleq c$ in the guard $g$,
\item $(1, x_i) \xra{~\fleq -c~} (1,x_0)$ for every constraint $x_i
  \ggeq c$ in the guard $g$
\end{enumerate}
\paragraph{Reset edges:}
\begin{enumerate}
  \setcounter{enumi}{4}
\item $(1, x_i) \xra{~\le 0~} (2,x_i)~$ and $~(2,x_i) \xra{~\le 0~}
  (1,x_i)$ for all $x_i \notin R$,
\item $(1, x_i) \xra{~\le 0~} (2, x_i)$ for all $x_i \in R$
\item $(2, x_0) \xra{~\le 0~} (0, x_0)~$ and $~(2,x_0) \xra{~\le 0~}
  (0, x_0)$ for all $x_i \in R$.
\end{enumerate}
\end{definition}
The edges between the first two columns of $G_t$ describe the changes
due to time elapse and constraints due to the guard. The edges between
the second and the third column represent the constraints arising due
to reset. Note that the only non-zero weights in the graph come from
the guard.

The following lemma establishes that the definition of a
transformation graph that we have given indeed reflects a transition
according to Definition~\ref{def:reflecting-a-transition}. The proof
is by direct verification. 

\begin{lemma}\label{lem:tgraph-reflects-t}
  The transformation graph $G_t$ of a transition $t$ reflects
  transition $t$.
\end{lemma}

Now that we have defined the transformation graph for a transition, we
will extend it to a transition sequence using a composition operator.

\begin{definition}
  Given two transformation graphs $G_1$ and $G_2$ we define its
  composition $G=G_1\odot G_2$. Supposing that the number of columns
  in $G_1$ and $G_2$ is $k_1$ and $k_2$, respectively; $G$ will have
  $k_1+k_2$-columns. Vertices of $G$
  are $\set{0,\dots,k_1+k_2-1}\times X$. The edge between vertices $(i,x)$ and
  $(j,y)$ exists and is the same as in $G_1$ if $i,j< k_1$; it is
  the same as between $(i-k_1,x)$ and $(j-k_1,y)$ in $G_2$ if
  $i,j\geq k_1$.  Additionally we add edges of weight $(\leq,0)$ from
  $(k_1-1,x)$ to $(k_1,x)$ and from $(k_1,x)$ to $(k_1-1,x)$.
\end{definition}

\noindent\includegraphics[scale=.28]{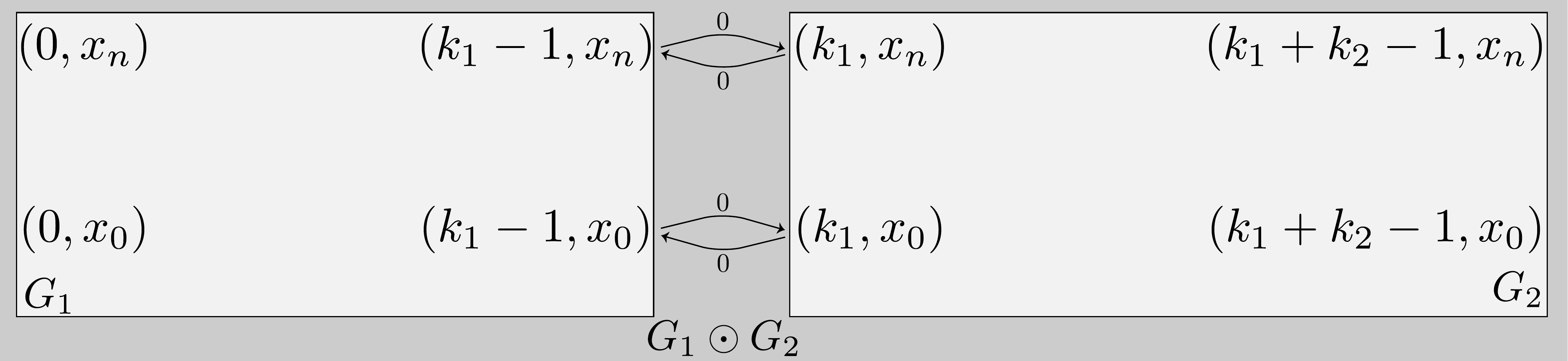}

By definition of composition, we get the following lemma. 
\begin{lemma}
  The composition is associative.
\end{lemma}

The following lemma is instrumental in lifting the definition of
transformation graph from a transition to a transition sequence. 

\begin{lemma}
  \label{lem:trans-graph-composition}
  If $G_1$, $G_2$ are transition graphs reflecting $\s_1$ and $\s_2$
  respectively, then $G_1 \odot G_2$ reflects their concatenation
  $\s_1\s_2$.
\end{lemma}

The transformation graph for a sequence of transitions is therefore
defined using the composition operation on graphs.

\begin{definition}
  The transformation graph $G_\s$ for a transition sequence $\s:= \xra{t_1}
  \xra{t_2} \dots \xra{t_k}$ is given by $G_1 \odot G_2 \cdots \odot
  G_k$ where $G_i$ is the transformation graph of $t_i$.
\end{definition}

From Lemma~\ref{lem:trans-graph-composition}, we get the following
property.
\begin{lemma}
  For every sequence of transitions $\s$ the graph $G_\s$ reflects
  $\s$.
\end{lemma}

We will make use of the transformation graph $G_\s$
to check if $\s$ is $\w$-iteratable. The following two corollaries are a first step in
this direction. They use Lemma~\ref{lem:negative-cycles} and
Lemma~\ref{lem:w-iteratability} to characterize iteratability in terms
of transformation graphs.
 
\begin{corollary}
  A sequence of transitions $\s$ is executable iff $G_\s$ does not have
  a negative cycle.
\end{corollary}

\begin{corollary}
  A sequence of transitions $\s$ is $\w$-iteratable iff for every
  $n=1,2,\dots$ the $n$-th fold composition $(G_\s\odot\cdots\odot
  G_\s)$ of $G_\s$ does not have a negative cycle.
\end{corollary}

Pick a transformation graph $G$ with no negative cycles. It is said to
be in \emph{canonical form} if the shortest path from some vertex $x$
to another vertex $y$ is given by the edge $x \to y$
itself. Floyd-Warshall's all pairs shortest paths algorithm can be
used to compute the canonical form. The complexity of this algorithm
is cubic in the number of vertices in the graph. As the transformation
graphs have many vertices, and the number of vertices grows each time
we perform a composition, we would like to work with smaller
transformation graphs.

\begin{definition}
  A \emph{short transformation graph}, denoted $|G|$, is obtained from
  a transformation graph $G$ without negative cycles, by first putting
  $G$ into canonical form and then restricting it to the leftmost and
  the rightmost columns.
\end{definition}


To be able to reason about $\s$ from $|G_\s|$, we need the following
lemma which says that $|G_\s|$ reflects $\s$ as well.
\begin{lemma}\label{lem:short-trans-graph-reflects}
  Let $G$ be a transformation graph with no negative cycles. Every
  solution to the short transformation graph $|G|$ extends to a
  solution of $G$. So if $\s$ is a sequence of transitions, $|G_\s|$
  reflects $\s$ too.
\end{lemma}

The purpose of defining short transformation graphs is to be able to
compute negative cycles in long concatenations of $\s$
efficiently. The following lemma is important in this regard, as it
will allow us to maintain only short transformation graphs during
each stage in the computation.

\begin{lemma}
\label{lem:composition-short}
For two transformation graphs $G_1$, $G_2$ without negative cycles, we
have: $||G_1|\odot|G_2||=|G_1\odot G_2|$
\end{lemma}


We finish this section with a convenient notation and a reformulation
of the above lemma that will make it easy to work with short
transformation graphs.

\begin{definition}\label{def:trans-to-short}
  $G_1\cdot G_2=|G_1\odot G_2|$.
\end{definition}

 
\begin{lemma}
  $G_1\cdot G_2=|G_1|\cdot|G_2|$. In particular, operation $\cdot$ is
  associative. 
\end{lemma}

Based on Definition~\ref{def:trans-to-short}, the transformation graph
$G_\s$ for a transition sequence $\s$ is in fact a two column graph
(left column with variables of the form $(0, x)$ and the right column
with variables of the form $(1,x)$). In the next section, we will use
this two column graph $G_\s$ to reason about the transition sequence
$\s$. 


\section{A pattern making $\w$-iteration impossible}
\label{sec:pattern-makes-w}

The effect of the sequence of timed transitions is fully described by
its transformation graph. We will now define a notion of a pattern in
a transformation graph that characterises those sequences of
transitions that cannot be $\w$-iterated. This characterization gives
directly an algorithm for checking $\w$-iterability.

For this section, fix a set of clocks $X$. We denote the number of
clocks in $X$ by $n$. Let $\s$ be a sequence of transitions $\s$, and
let $G_\s$ be the corresponding transformation graph. Let
$\lleft(G_\s)$ denote the zone that is the restriction of $G_\s$ to
the leftmost variables. Similarly for $\rright(G_\s)$, but for the
rightmost variables.  As $G_\s$ reflects $\s$, the zone $\lleft(G_\s)$
describes the maximal set of valuations that can execute $\s$
once. Similarly, when we consider the $i$-fold composition,
$(G_\s)^i$, the zone $\lleft((G_\s)^i)$ describes the set of
valuations that can execute $\s$ $i$-times.

We want to find a constant $\k$ such that if $\s$ is $\w$-iterable,
the left columns of $(G_\s)^\k$ and $(G_\s)^{2\k}$ are the same. Note
that as the number of regions is finite, we will eventually reach $i$
and $j$ such that $\lleft((G_\s)^i)$ and $\lleft((G_\s)^j)$ intersect
the same set of regions. Moreover as the left column is not
increasing, we would eventually get a $\k$ that we want. However with
this na\"{i}ve approach we get a bound $\k$ which is even exponential
in the number of regions.


We will show that if $\s$ is $\w$-iterable then $\k\le n^2$. The trick
is to study shortest paths in the graph $(G_\s)^i$. To this regard, we
will look at paths in this composition as trees, which we call
\emph{p-trees}. We begin with an illustration of a p-tree in
Figure~\ref{fig:p-tree}. It explains the definition of a p-tree and
shows the one-to-one correspondence between paths in a composition of
$G_\s$ and $p$-trees. The weight of a path is the sum of weights in
the tree.

\begin{figure}[htbp]
  \centering
  \includegraphics[scale=.5]{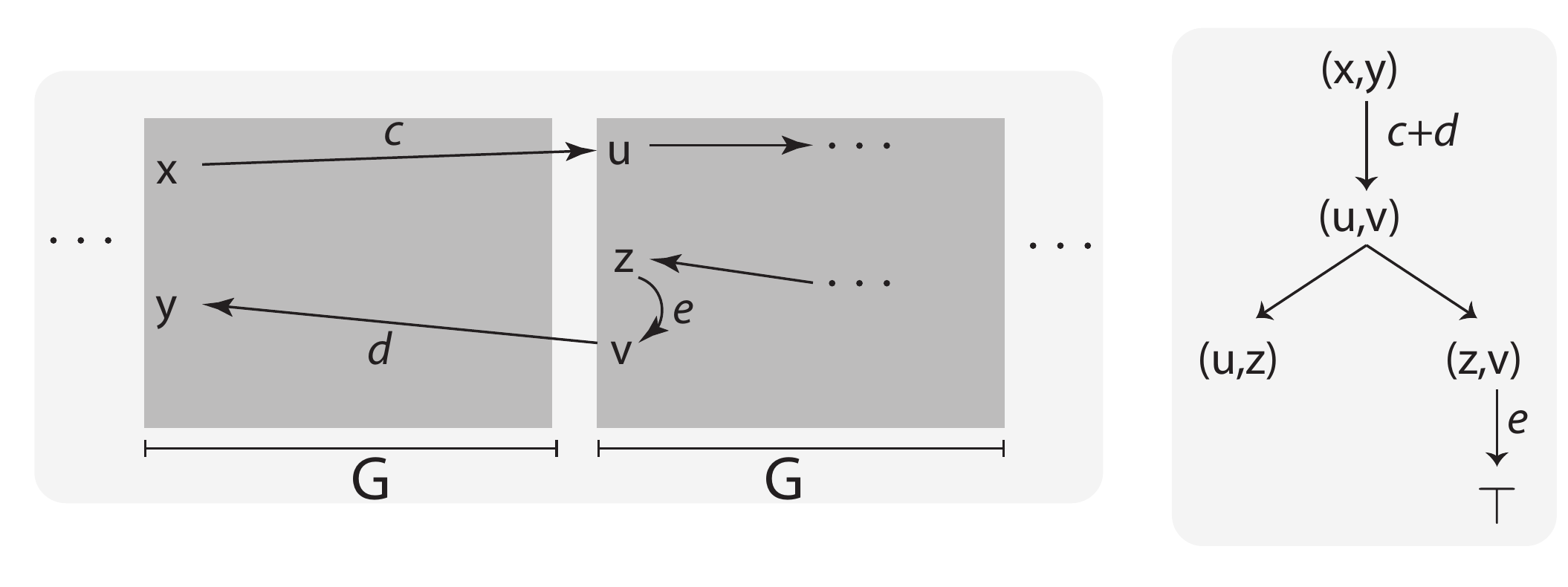}
  \caption{A path in a composition of $G$'s and the corresponding
    p-tree.}
  \label{fig:p-tree}
\end{figure}

\begin{definition}
  A \emph{p-tree} is a weighted tree whose nodes are labelled with
  $\top$ or pairs of variables from $X$. Nodes labelled by $\top$ are
  leaves. A node labelled by a pair of variables $(x,y)$ can have one
  or two children as given in
  Figure~\ref{fig:p-tree-defn}. 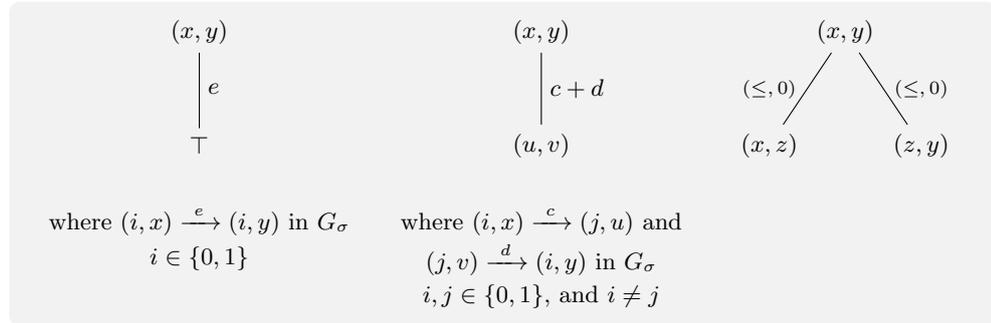
\begin{figure}[!h]
\centering
\begin{tikzpicture}

\fill [rounded corners, gray!10] (-2.5, -3.4) rectangle (10.5,0.9);
\node (s0) at (0,0.5) {$(x, y)$};
\node (s1) at (0,-1) {$\top$};
\draw (s0) -- (s1);
\node [right] at (0, -0.25) {\footnotesize $e$};
\node at (0, -2) {where $(i,x) \xra{~e~} (i,y)$ in $G_\s$};
\node at (0, -2.5) {$i\in\set{0,1}$};
\begin{scope}[xshift=4.5cm]
\node (s0) at (0,0.5) {$(x, y)$};
\node (s1) at (0,-1) {$(u, v)$};
\draw (s0) -- (s1);
\node [right] at (0, -0.25) {\footnotesize $c+d$};
\node at (0, -2) {where $(i,x) \xra{~c~} (j,u)$ and};
\node at (0, -2.5) {$(j,v) \xra{~d~} (i,y)$ in $G_\s$};
\node at (0, -3) {$i,j\in\set{0,1}$, and $i\not=j$};
\end{scope}

\begin{scope}[xshift=8.5cm]
\node (s0) at (0,0.5) {$(x, y)$};
\node (s1) at (-1,-1) {$(x, z)$};
\node (s2) at (1,-1) {$(z, y)$};
\draw (s0) -- (s1);
\draw (s0) -- (s2);
\node at (-1, -0.25) {\scriptsize $(\le,0)$};
\node at (1, -0.25) {\scriptsize $(\le,0)$};

\end{scope}
\end{tikzpicture}
\caption{Possible children of $(x, y)$ in a p-tree}
\label{fig:p-tree-defn}
\end{figure}

  The \emph{weight of a p-tree} is the sum of weights that appear in
  it.  A $p$-tree is \emph{complete} if all the leaves are labelled by
  $\top$.

  An $(x,y)$-$(u,v)$ context is a p-tree with the root labelled
  $(x,y)$ and all the leaves labelled $\top$ except for one leaf that
  is labelled $(u,v)$.
\end{definition}

The definition of p-tree reflects the fact that each time the path can
either go to the left or to the right, or stay in the same column.
Every p-tree represents a path in an $i$-fold composition $(G_\s)^i$
and every path in $(G_\s)^i$ can be seen as a p-tree.

\begin{lemma}
  Take an $i\in \set{1,2,\dots}$ and consider an $i$-fold composition
  of $G_\s$.  For every pair of vertices $(j,x)$ to $(j,y)$ in the
  same column of $G_\s^{i}$: if there is a path of weight $w$ from
  $(j,x)$ to $(j,y)$ then there is a complete p-tree of weight $w$
  with the root labeled $(x,y)$.
\end{lemma}
\begin{proof}
  Consider a path $\pi$ from $(j,x)$ to $(j, y)$. We construct its
  p-tree by an induction on the length of $\pi$. If $\pi$ is just an
  edge $(j, x) \xra{~e~} (j,y)$, then the corresponding p-tree would
  have $(x,y)$ as root and a single child labeled $\top$. The weight
  of the edge would be $e$.

  Suppose $\pi$ is a sequence of edges. We make the following division
  based on its shape.
  
  \paragraph*{Case 1:} The path $\pi$ could cross the $j^{th}$ column
  at some vertex $(j,z)$ before reaching $(j, y)$: $(j,x) \xra{}
  \cdots \xra{} (j,z) \xra{} \cdots \xra{} (j,y)$. By induction, there
  are complete p-trees for the paths $(j,x) \xra{} \cdots (j, z)$ and
  $(j,z) \xra{} \dots (j,y)$. The p-tree for $\pi$ would have $(x,y)$
  as root and $(x, z)$ and $(z,y)$ as its two children. The edge
  weights to the two children would be $(\le,0)$.  The p-trees of the
  smaller paths would be rooted at $(x,z)$ and $(z,y)$.
  
  \paragraph*{Case 2:} The path $\pi$ never crosses the $j^{th}$
  column before reaching $(j,y)$. This means the path is entirely to
  either the left or right of the $j^{th}$ column. Let us assume it is
  to the right. The case for left is similar. So $\pi$ looks like:
  $(j,x) \xra{c} (j+1, u) \cdots (j+1, v) \xra{d} (j,y)$. Note that
  $u$ cannot be equal to $v$ since the path of the minimal weight must
  be simple (in the case when the path has exactly three nodes $(j,x),
  (j+1,v), (j,y)$, there is actually a direct edge from $(j,x)$ to
  $(j,y)$, as $|G_\s|$ is in the canonical form). By induction, the
  smaller path segment from $(j+1, u)$ to $(j+1, v)$ has a p-tree. The
  p-tree for $\pi$ would have $(x,y)$ as root and a single child
  $(u,v)$. The weight of this edge would be $c+d$. The p-tree for the
  smaller path would be rooted at $(u,v)$.
\end{proof}

\begin{lemma}\label{lem:p-tree-to-path}
  If there is a complete p-tree with the root $(x,y)$, weight $w$ and
  height $k$ then in $G_\s^{2k}$ the weight of the shortest path
  from $(k,x)$ to $(k,y)$ is at most $w$.
\end{lemma}
\begin{proof}
  For every level moved down in the p-tree, the corresponding path
  either moves one column left or one column right. As the height of
  the p-tree is bounded by $k$, there is a path that spans less than
  $k$ columns. This path goes from $(k, x)$ to $(k,y)$ in the
  composition $G_\s^{2k}$. Its weight is the weight of the
  tree. Therefore the smallest weight of a path from $(k,x)$ to
  $(k,y)$ is at most $w$.
\end{proof}

Now that we have established the correspondence between paths and
p-trees, we are in a position to define a \emph{pattern} in the paths
that causes non-iterability.

\begin{definition}
  A \emph{pattern} is an $(x,y)$-$(x,y)$ context for some variables
  $x$ and $y$, whose weight is negative
  (c.f. Figure~\ref{fig:pattern}).  We say that $\s$ \emph{admits} a
  pattern if there is a pattern (notice that the definition of p-tree
  depends on $\s$).
\end{definition}

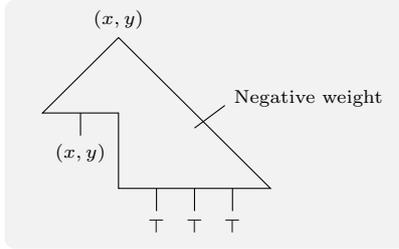
\begin{figure}
\centering

\begin{tikzpicture}


\fill [rounded corners, gray!10] (-1.5, -2.8) rectangle (3.8, 0.5);
\draw (0,0) -- (-1,-1) -- (0,-1) -- (0, -2) -- (2, -2) -- (0,0);
\draw (-0.5, -1) -- (-0.5, -1.3);
\draw (0.5, -2) -- (0.5, -2.3);
\draw (1, -2) -- (1, -2.3);
\draw (1.5, -2) -- (1.5, -2.3);

\node [above] at (0,0) {\scriptsize $(x,y)$};
\node [below] at (-0.5, -1.3) {\scriptsize $(x,y)$};
\node [below] at (0.5, -2.3) {\scriptsize $\top$};
\node [below] at (1, -2.3) {\scriptsize $\top$};
\node [below] at (1.5, -2.3) {\scriptsize $\top$};

\node at (2.5, -0.8) {\scriptsize Negative weight};
\draw [thin] (1, -1.2) -- (1.4, -.9);

\end{tikzpicture}
\caption{Pattern: root is $(x,y)$ and the only leaf that is not $\top$
  is $(x,y)$ again. Moreover, the weight has to be negative.}
\label{fig:pattern}
\end{figure}


The next proposition says that patterns characterize
$\w$-iterability. Morever, it ensures that we can conclude after $n^2$
iterations.  For transformation graphs $G_1, G_2$, we will write
$G_1\bumpeq G_2$ if $\lleft(G_1)=\lleft(G_2)$ and
$\rright(G_1)=\rright(G_2)$.

\begin{proposition}\label{prop:pattern-characterization}
  If $\s$ admits a pattern then there is no valuation from which $\s$
  is $\w$-iterable.  If $\s$ does not admit a pattern, then for every
  $i=1,2\dots$ we have $(G_\s)^{n^2}\bumpeq(G_\s)^{n^2+i}$.
\end{proposition}

Proof of the above proposition follows from
Lemmas~\ref{lem:pattern-no-iteration}, \ref{lem:no-pattern-bounded-height},
\ref{lem:above-2n2-same}.  

\begin{lemma}\label{lem:pattern-no-iteration}
  If $\s$ admits a pattern then there is no valuation from which $\s$
  is $\w$-iteratable.
\end{lemma}
\begin{proof}
  We will show that there is $k$ such that there is a negative cycle
  in the $k$-th fold composition $G_\s\odot\cdots \odot G_\s$.

  Suppose $\s$ admits a pattern.  Suppose we have clocks $x,y$ such that
  $x\leq y$ is implied by $\s$. This means that the last reset of $y$
  happens before the last reset of $x$ in $\s$. In consequence $x\leq
  y$ is an invariant at the end of every iteration of $\s$.

  Consider a sufficiently long composition $G_\s\odot\cdots \odot
  G_\s$. Since there is a pattern, for some $i<j$ we have one of the two cases
  \begin{enumerate}
  \item there is a path $(i,x)$ to $(j,x)$ and from $(j,y)$ to $(i,y)$
    whose sum of weights is negative.
  \item there is a path $(j,x)$ to $(i,x)$ and from $(i,y)$ to $(j,y)$
    whose sum of weights is negative.
  \end{enumerate}
  Observe that since $x\leq y$, we have an edge of weight at most
  $(\leq,0)$ from $y$ to $x$ in every iteration.

  The case 1 makes the edge $y\to x$ more and more negative when going
  to the right. The case 2 makes the leftmost edge $x\leq y$ more and
  more negative depending on the number of iterations.

  The second case clearly implies that an infinite iteration is
  impossible: in every valuation permitting $\w$-iteration the
  distance from $x$ to $y$ would need to be infinity.

  The first case tells that in subsequent iterations the difference
  between $x$ and $y$ should grow. Since every variable is reset in
  every iteration, this implies that the amount of time elapsed in
  each iteration should grow too. But this is impossible since, as we
  will show in the next paragraph, the presence of the edge $(j,y)$ to
  $(i,y)$ of weight $d$ implies that the value of $y$ after each
  execution of $\s$ is bounded by $d$. In consequence the difference
  between $x$ and $y$ is bounded too. A contradiction.

  \begin{figure}[tbh]
    \centering
    \includegraphics[scale=.3]{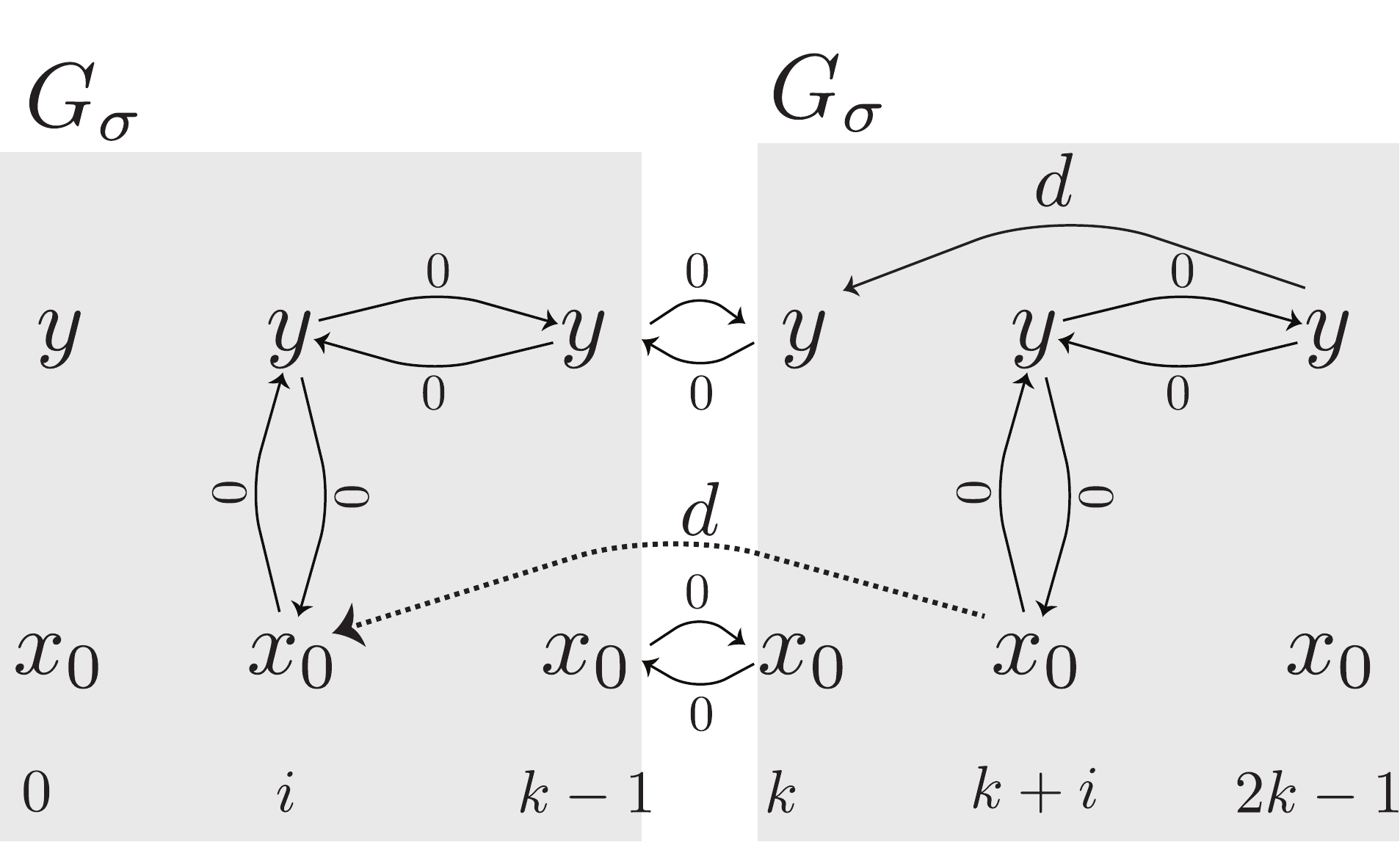}
    \caption{An edge from $(2,y)$ to $(1,y)$ gives a bound on time elapse}
    \label{fig:d-edge}
  \end{figure}
  It remains to see why the edge $(j,y)$ to $(i,y)$ of weight $d$
  implies that the value of $y$ at the end of each iteration of $\s$
  is bounded. For simplicity of notation take $j=1$ and $i=2$, the
  argument is the same for other values. We will show that the last
  two resets of $y$ in the consecutive executions of $\s$ need to
  happen in not more than $d$ time units apart. This implies that the
  value of $y$ is bounded by $d$.

  In Figure~\ref{fig:d-edge} we have pictured the composition of two
  copies of the transition graph $G_\s$. The graph has $k$ columns:
  numbered from $0$ to $k-1$. The last reset of $y$ is in the column
  $i$, so in the second copy it is in the column $k+i$. The black
  edges with weight $0$ come from the definition of the transition
  graph. For example the horizontal weight $0$ edges between $y$'s
  are due to the fact that $y$ is not reset between columns $i$ and
  $k-1$. As we can see from the picture, the edge of weight $d$ from
  $(k,y)$ and $(2k+1,y)$ induces an edge of weight $d$ from
  $(k+i,x_0)$ to $(i,x_0)$. Take a valuation $v$ satisfying the
  pictured composition of the two graphs $G_\s$. The induced edge
  gives us $v(i,x_0)-v(k+i,x_0)\leq d$. Recall that the value
  $-v(i,x_0)$ represents the time instance when the column $i$
  ``happens''. Rewriting the last inequality as
  $(-v(k+i,x_0))-(-v(i,x_0))\leq d$, we can see that between columns
  $i$ and $k+i$ at most $d$ units of time have passed. So the value
  $v(k-1,y)-v(k-1,x_0)$, that is the value of $y$ at the end of the
  first iteration of $\s$ is bounded by $d$.\qed
\end{proof}

The above lemma says that if there is a pattern, then $\s$ cannot be
iterated. We will now show that if there is no pattern, then the
p-trees representing shortest paths are bounded by $n^2$ and hence
$n^2$ iterations will be sufficient to conclude if the sequence is
$\w$-iterable. 

\begin{lemma}\label{lem:no-pattern-bounded-height}
  If $\s$ does not admit a pattern then for every $(x,y)$ there is a
  complete p-tree of minimal weight and height bounded by $n^2$ whose
  root is $(x,y)$.
\end{lemma}
\begin{proof}
  In order to get a p-tree with minimal weight, observe that if we
  have a repetition of a label in the tree then we can cut out the
  part of the tree between the two repetitions. Moreover, we know that
  this part of the tree would sum up to a non-negative weight as by
  assumption $\s$ does not have a pattern. Therefore, cutting out the
  part between repetitions gives another tree that has smaller height
  and does not have a bigger weight. The height of a p-tree with no
  repetitions is bounded by $n^2$.
\end{proof}

\begin{lemma}\label{lem:above-2n2-same}
  If $\s$ does not admit a pattern, then for every $i=1,2\dots$ we have
  $G_\s^{n^2+i}\bumpeq G_\s^{n^2}$.
\end{lemma}
\begin{proof}
  Due to Lemma~\ref{lem:no-pattern-bounded-height}, the
  shortest paths between any two variables in the
  leftmost column does not cross $n^2$ columns. Similarly the shortest
  path between any two variables in the rightmost column cannot go
  more than $n^2$ columns to the left. As this value is the same for
  both $G_\s^{n^2}$ and $G_\s^{n^2+i}$, the lemma follows. \qed
\end{proof}

Based on Proposition~\ref{prop:pre-process}, we get a procedure for
checking if $\s$ is $\w$-iterable.

\begin{theorem}\label{thm:iteratability-test}
  Let $\s$ be a sequence of transitions and let $n$ be the number of
  clocks. Following is a procedure for checking if $\s$ is
  $\w$-iterable.

  \begin{enumerate}
  \item If $\s$ satisfies Lemma~\ref{lem:conditions-non-iteratable},
    report $\s$ is not $\w$-iterable. Otherwise, continue.
  \item If there is no clock that is both reset and is checked for
    $\ggeq c$ with $c > 0$ in some guard, report $\s$ is iterable.  If
    there is such a clock, remove from $\s$ all guards containing
    clocks that are not reset and continue.
  \item Compute $G^1=G_\s$; stop if it defines the empty relation.
  \item Iteratively compute $G^{2^{k+1}}=G^{2^k}\cdot G^{2^k}$; stop
    if the result defines the empty relation, or $2^k>n^2$, or
    $G^{2^{k+1}}\bumpeq G^{2^k}$.
  \end{enumerate}
  If a result is not defined or $2^k>n^2$ then $\s$ is not
  $\w$-iterable.  If $G^{2^{k+1}}=G^{2^k}$ then $\lleft(G^{2^k})$ is
  the zone consisting precisely of all the valuations from which $\s$
  is $\w$-iterable.
\end{theorem}
\begin{proof}
  The first two steps are justified by
  Proposition~\ref{prop:pre-process}. We discuss the third and fourth
  steps.
  
 If at some moment the result of an operation is not defined then
  there is a negative cycle in $G^{2^k}_\s$, so $\s$ can be
  executed not more than $2^k$ times. Hence $\s$ is not
  $\w$-iteratable. 

   If $2^k>n^2$ then there is a pattern and $\s$ is not
  $\w$-iteratable (c.f.\ Lemma~\ref{lem:above-2n2-same}).

  If $G^{2^{k+1}} = G^{2^k}$, then it means that the left column will
  not change on further compositions. As $G_\s$ is the graph that
  reflects $\s$,  we have that
  $\lleft(G^{2^k})$ is the set of valuations from which $\s$ is
  $\w$-iteratable. 
  \qed

\end{proof}

Complexity of this procedure is $\Oo((|\s| + \log n) \cdot n^3)$. The
graph $G_\s$ is build by incremental composition of the transitions in
$\s$. Each transition in $\s$ can be encoded as a transformation graph
over $2n$ variables. The sequential composition of two transformation
graphs over $2n$ variables uses a DBM over $3n$ variables (with $n$
variables shared by the two graphs). The canonicalization of this DBM
is achieved in time $(3n)^3$ and yields a transformation graph over
$2n$ variables corresponding to the composition of the
relations. Hence, $G_\s$ is obtained in time $\Oo(2 |\s| \cdot
(3n)^3)$ assuming that each step in $\s$ corresponds to a transition
followed by a delay. Once $G_\s$ has been computed, $(G_\s)^{n^2}$ is
obtained in time $\Oo(2 \log n \cdot (3n^3))$ using the fact that
$(G_s)^{2k} = (G_s)^k \cdot (G_s)^k$ for $k \ge 1$.



\section{Experiments}
\label{sec:algorithm}
In this section we give some indications about the usefulness of the
$\w$-iterability check. The example from page~\pageref{fig:example}
shows that the gains from our $\w$-iteration procedure can be
arbitrarily big. Still, it is more interesting to see the performance
of $\w$-iterability check on standard benchmark models. For this we
have taken the collection of the same standard models as in many other
papers on the verification of timed systems and in particular in~\cite{Laarman:CAV:2013}.

Our testing scheme is as follows: we will verify properties given by
B\"uchi automata on the standard benchmarks. To do this, we take the
product of the model and property automata and check for B\"uchi
non-emptiness on this product automaton. Algorithms for the B\"uchi
non-emptiness problem can broadly be classified into two kinds: nested
DFS-based~\cite{Holzmann:DIMACS:1997} and Tarjan-based decomposition
into SCCs~\cite{Couvreur:FM:1999}. Both these algorithms are
essentially extensions of the simple DFS algorithm with extra
procedures that help identify cycles containing accepting
states. Currently, no algorithm is known to outperform the rest in all
cases \cite{GS:MEMICS:2009}.

\paragraph*{Restricting to weak B\"uchi properties.} In order to focus
on the influence of iterability checking we consider weak B\"uchi
properties (all states in a cycle are either accepting or
non-accepting). In this case, the simple modification of DFS is the
undisputed best algorithm: to find an accepting loop, it is enough to
look for it on the active path of the DFS~\cite{Cerna:MFCS:2003}. In
other words, it is nested DFS where the secondary DFS search is never started.

The algorithm, that we will call DFS with subsumption (DFSS), performs
a classical depth-first search on the finite abstracted zone graph
$\ZG^{\extraLUp}(\Aa)$. It uses the additional information provided by
the zones to limit the seach from the current node $(q',Z')$ in two
cases: (1)~if $q'$ is accepting and on the stack there is a path
$\sigma$ from a node $(q',Z'')$ to the current node with $Z'' \incl
Z'$ then report existence of an accepting path; and (2)~if there is a
fully explored a node $(q',Z'')$, i.e a node not on the stack, with
$Z' \incl Z''$ then ignore the current node and return from the DFS
call. This is the algorithm one obtains after specialisation of the
algorithm of Laarman \textit{et al.}~\cite{Laarman:CAV:2013} to weak
Buchi properties. It is presented in full in
Appendix~\ref{sec:appendix:algo}.

Now, it is quite clear how to add $\w$-iterability check to this
algorithm. We extend case~(1) above using $\w$-iterability check. When
$q'$ is accepting but $Z'' \not \incl Z'$ we use $\w$-iterability
check on $\sigma$. If $\sigma$ is iterable from $(q',Z')$, we have
detected an accepting path and we stop the search. We refer to this
algorithm as iDFSS: DFSS with $\w$-iteration testing. Notice that when
iDFSS stops thanks to $\w$-iteration check, DFSS would just continue
its search. The algorithm is listed in 
Appendix~\ref{sec:appendix:algo}.

Our aim is to test the gains of iDFSS over DFSS. One can immediately
see that iDFSS is not better than DFSS if there is no
accepting path. Therefore in our examples we consider only the
cases when there is an accepting path, to see if iDFSS has an effect
now.

We tried variants of properties from \cite{Laarman:CAV:2013} on the
standard models (TrainGate, Fischer, FDDI and CSMA/CD). The results
were not very encouraging: $\w$-iterability check was hardly ever
used, and when used the gains were negligible. Based on a closer look
at these timed models, we argue that while these models are
representative for reachability checking, their structure is not well
adapted to evaluate algorithms for B\"uchi properties. We explain this in
more detail below.

\paragraph{A note on standard benchmarks.}
In three out of four models (Fischer, TrainGate and CSMA/CD), on every
loop there is a true zone (which is the zone representing the set of
all possible valuations). Moreover, in Fischer and TrainGate, a big
majority of configurations have true zones, and even more strikingly,
the longest sequence of configurations with a zone other than true
does not exceed the number of components in the system: for example,
in Fisher-$3$ there are at most $3$ consecutive nontrivial zones. As
we explain in the next paragraph, in the case of CSMA/CD all the loops
turn out to be almost trivial. Thus in all the three models one can as
well ignore timing information for loop detection: it is enough to
look at configurations with the true zone. This analysis explains the
conclusion from~\cite{Laarman:CAV:2013} where it is reported that
checking for counterexamples is almost instantaneous. Indeed, checking
for simple untimed weak B\"uchi properties on these models will be
very fast since every repetition of a state $q$ of the automaton will
give a loop that is $\w$-iteratable in $\ZG^{\extraLUp}(\Aa)$. This
loop will be successfully detected by the inclusion test~(1) in DFSS
algorithm since both zones will be true.  The fourth model from
standard benchmarks - FDDI - is also very peculiar. In this paper we
are using $\extraLUp$ abstraction for easy of comparison. Yet more
powerful abstractions allow to eliminate time component completely
from the model~\cite{Herbreteau:CAV:2013}. This indicates that
dependencies between clocks in the model are quite weak.

\begin{figure}[t]
  \centering
  \mbox{%
    \begin{minipage}{.49\textwidth}
      \begin{tikzpicture}

  \path[every node/.style = {shape=circle,draw,minimum size=1.1cm},
  thick, inner sep=2pt,font=\tiny]
  node (wait)   at (0,0) {$WAIT$}
  node (start)  at (4,0)
  {$\begin{array}{c}
      START\\
      (x_i \le L)\\
    \end{array}$}
  node (retry)  at (2,-1.8)
  {$\begin{array}{c}
      RETRY\\
      (x_i < 2S)\\
      \end{array}$};
  
  \path[->,line width=1pt,font=\tiny]
  ([xshift=-0.25cm]wait.west) edge (wait)
  (wait)  edge[in=80,out=100,loop]  node[above] {$cd_i, \{x_i\}$}
  (wait)
  (wait)  edge              node[above] {$begin_i, \{x_i\}$} (start)
  (wait)  edge              node[right]
  {$\begin{array}{c}
      busy_i\\
      \{x_i\}\\
      \end{array}$} (retry)
  (wait)  edge[bend right=15]  node[left]  {$cd_i, \{x_i\}$} (retry)
  (retry) edge[in=180,out=200,loop] node[left]
  {$cd_i, \{x_i\}$} (retry)
  (retry) edge[in=0,out=340,loop,dashed]  node[right]
  {$\mathbf{busy_i, \{x_i\}}$} (retry)
  (retry) edge[bend right=15]  node[right]
  {$begin_i, \{x_i\}$} (start)
  (start) edge                node[left,pos=0.3]
  {$\begin{array}{c}
      (x_i<S)\\
      cd_i\\
      \{x_i\}
      \end{array}$} (retry)
  (start) edge[bend right=30]  node[above]
  {$(x_i=L), end_i, \{x_i\}$} (wait);
  
\end{tikzpicture}

    \end{minipage}
    \hfill
    \begin{minipage}{.49\textwidth}
      \begin{tikzpicture}

  \path[every node/.style = {shape=circle,draw,minimum size=1.1cm},
  thick, inner sep=2pt,font=\tiny]
  node (idle)       at (0,0)    {$IDLE$}
  node (busy)       at (4,0)    {$BUSY$}
  node (collision)  at (2,-1.8)
  {$\begin{array}{c}
      COLL.\\
      (y < S)\\
    \end{array}$};
  
  \path[->,line width=1pt,font=\tiny]
  ([xshift=-0.25cm]idle.west) edge (idle)
  (idle) edge[bend right=15]
  node[below] {$begin_i, \{y\}$} (busy)
  (busy) edge[in=80,out=100,loop]
  node[above] {$(y \ge S), busy_i$} (busy)
  (busy) edge[bend right=15]
  node[above] {$end_i, \{y\}$} (idle)
  (busy) edge[bend left=15]
  node[right] {$\begin{array}{c}
      (y < S)\\ begin_i\\ \{y\}\end{array}$} (collision)
  (collision) edge[bend left=15]
  node[left] {$\begin{array}{c}
      (y < S)\\ \{cd_1,\dots,cd_N\}\\ \{y\}
    \end{array}$} (idle);
  
\end{tikzpicture}

    \end{minipage}
  }
  \caption{Model of the CSMA/CD protocol: station (left) and bus
    (right).}
  \label{fig:CSMA_CD}
\end{figure}
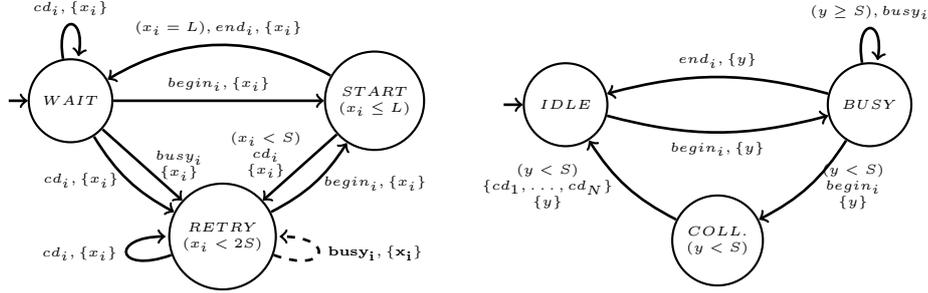

The case of CSMA/CD gives an interesting motivation for testing
B\"uchi properties. CSMA/CD is a protocol used to communicate over a
bus where multiple stations may emit at the same time. As a result,
message collisions may occur and have to be
detected. Figure~\ref{fig:CSMA_CD} shows a model for this protocol. We
assume that the stations and the bus synchronize on actions. In
particular, the transition from state $COLLISION$ to state $IDLE$ in
the bus automaton synchronizes on action $cd_i$ in all $N$
stations. The model is parametrized by the propagation delay $S$ on
the bus and the time $L$ needed to transmit a message which are
usually set to $L=808$ and
$S=26$~\cite{Tripakis:FMSD:2001,UPPAAL-CSMACD}.

It turns out that the $busy_i$ loop on state $RETRY$ in the station
automaton is missing in the widely used
model~\cite{Tripakis:FMSD:2001,UPPAAL-CSMACD}. In consequence, in this
model there is no execution with infinitely many collisions and
completed emissions. Even more, once some process enters in a
collision, no process can send a message afeterwards. 
This
example confirms once more that timed models are compact descriptions
of complicated behaviours due to both parallelism and interaction
between clocks. B\"uchi properties can be extremely useful in making
sure that a model works as intended: the missing behaviors can be
detected by checking if there is a run where every collision is
followed by an emission. Adding the $busy_i$ loop on state $RETRY$
enables interesting behaviours where the stations have collisions and
then they restart sending messages. The resulting model has
significantly more reachable states and behaviours (see
Appendix~\ref{sec:models}).

The other issue with CSMA/CD is that it has Zeno behaviours, and they
appear precisely in the interesting part concerning collision. A
solution we propose is to enforce $(y \ge 1)$ on all transitions from
state $BUSY$ in the bus automaton. By chance, the modified model does
not have new reachable states. But of course now all loops are nonZeno.

To sum up the above discussion, due to the form of the benchmark models, we are lead to
consider B\"uchi properties that refer to time. This gives us a
product automaton having zones with non-trivial interplay of clocks,
and consequently making the B\"uchi non-emptiness problem all the more
challenging. Figure~\ref{fig:benchmarks} (right) presents the property
we have checked on the modified CSMA/CD model. It checks if station
$1$ can try to transmit fast enough, and if it arrives to send a message, the delay is
not too long. The properties we have considered for other models are
listed in Appendix~\ref{sec:models}. 
An important point is that since the
exploration ends as soon as it finds a loop, the order of exploration
influences the results. For fairnesss of the comparison we have run
numerous times the two algorithms using random exploration order, both
algorithms following the same order each time.  The results of the
experiments are presented in Figure~\ref{fig:benchmarks} (left). All
the examples have an accepting run.

\paragraph{Bottomline.} The table shows that there are quite natural
properties of standard timed models where the use of $\w$-iteration
makes a big difference. Of course, there are other examples
where $\w$-iteration contributes nothing. Yet, when it gives nothing,
$\w$-iteration will also not cost much because it will not be called
often. 

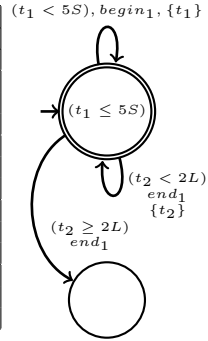
\begin{figure}[t]
  \centering
  \mbox{%
    \begin{minipage}{.74\textwidth}
      \scriptsize
      \begin{tabular}{|l|r|r|r|r|r|r|r|r|r|}
  \hline
  \multirow{3}{*}{Model} & \multicolumn{3}{|c|}{DFSS algorithm} &
  \multicolumn{6}{|c|}{iDFSS algorithm}\\
  \cline{2-10}
  & \multicolumn{3}{|c|}{Visited nodes} &
  \multicolumn{3}{|c|}{Visited nodes} &
  \multicolumn{3}{|c|}{Iterability checks}\\
  \cline{2-10}
  & Mean & Min & Max & Mean & Min & Max & Mean & Min & Max\\
  \hline
  CSMA/CD 4 & 9956 & 9699 & 10862 & 137 & 8 & 1646 & 1 & 1 & 1\\
  CSMA/CD 5 & 12108 & 11315 & 14433 & 219 & 9 & 1331 & 1 & 1 & 1\\
  CSMA/CD 6 & 16284 & 12931 & 25861 & 561 & 10 & 3523 & 1 & 1 & 1\\
  \hline
  Fischer 3 & 1067 & 275 & 9638 & 294 & 7 & 1625 & 3 & 1 & 28\\
  Fischer 4 & 6516 & 365 & 129211 & 1284 & 7 & 21806 & 12 & 1 & 1296\\
  Fischer 5 & 91124 & 455 & 1171755 & 17402 & 7 & 352682 & 30 & 1 &
  1131\\
  \hline
  FDDI 8 & 4807 & 1489 & 11778 & 1652 & 1105 & 4022 & 23 & 1 & 71\\
  FDDI 9 & 6326 & 2544 & 19046 & 1979 & 1247 & 11514 & 24 & 1 & 73\\
  FDDI 10 & 7156 & 2840 & 25705 & 2101 & 1390 & 11587 & 22 & 1 & 65\\
  \hline
  Train Gate 3 & 1017 & 212 & 3649 & 109 & 4 & 528 & 1 & 1 & 1\\
  Train Gate 4 & 14928 & 281 & 65358 & 1874 & 4 & 11844 & 3 & 1 & 90\\
  Train Gate 5 & 448662 & 350 & 1243873 & 63000 & 4 & 235957 & 29 &
  1 & 397\\
  \hline
\end{tabular}


    \end{minipage}
    \hfill
    \begin{minipage}{.25\textwidth}
      \begin{tikzpicture}

  \path[every node/.style = {shape=circle,draw,minimum size=1cm},
  thick, inner sep=2pt,font=\tiny]
  node[double] (q0)   at (0,0)    {$(t_1 \le 5S)$}
  node         (q1)   at (0,-2.5) {};
  
  \path[->,line width=1pt,font=\tiny]
  ([xshift=-0.25cm]q0.west) edge (q0)
  (q0) edge[in=80,out=100,loop]
  node[above] {$(t_1 < 5S), begin_1, \{t_1\}$}
  (q0)
  (q0) edge[in=265,out=285,loop]
  node[right] {$\begin{array}{c}
      (t_2 < 2L)\\ end_1\\ \{t_2\}\end{array}$}
  (q0)
  (q0) edge[bend right=60]
  node[right,pos=0.65] {$\begin{array}{c}
      (t_2 \ge 2L)\\ end_1 \end{array}$}
  (q1);
\end{tikzpicture}

    \end{minipage}
  }
  \caption{Benchmarks (left) and property checked on CSMA/CD (right).}
  \label{fig:benchmarks}
\end{figure}

\section{Conclusions}

We have presented a method for testing $\w$-iterability of a cycle in a
timed automaton. Concentrating on iterability, as opposed to loop
decomposition~\cite{ComJur99}, has allowed us to obtain a
relatively efficient procedure, with a complexity comparable to zone
canonicalisation. We have performed experiments on the usefulness of this
procedure for testing weak B\"uchi properties of standard benchmark
models. This has led us to examine in depth the structure of these
models. This structure explains why testing untimed B\"uchi
properties turned out to be immediate. When we have tested for timed
B\"uchi properties the situation changed completely, and we could
observe substantial gains on some examples. 

There is no particular difficulty in integrating our $\w$-iterability
test in a tool for checking all B\"uchi properties.  Due to the lack
of space we have chosen not to report on the prototype here. Let us
just mention that we prefer methods based on Tarjan's strongly
connected components since they adapt to multi B\"uchi properties, and
allow to handle Zeno issues in a much more effective way than strongly
non-Zeno construction~\cite{Herbreteau:ATVA:2010}.

Although efficient in some cases, $\w$-iterability test does not solve
all the problems of testing B\"uchi properties of timed automata.  In
particular when there is no accepting loop in the automaton, the test
brings nothing, and one can often observe a quick state blowup. New
ideas are very much needed here.


\bibliographystyle{plain} \bibliography{m.bib}

\begin{thebibliography}{10}

\bibitem{AD:TCS:1994}
R.~Alur and {D.L.} Dill.
\newblock A theory of timed automata.
\newblock {\em Theoretical Computer Science}, 126(2):183--235, 1994.

\bibitem{Beh05}
G.~Behrmann.
\newblock Distributed reachability analysis in timed automata.
\newblock {\em {STTT}}, 7(1):19--30, 2005.

\bibitem{Behrmann:STTT:2006}
G.~Behrmann, P.~Bouyer, K.~G. Larsen, and R.~Pel{\'a}nek.
\newblock Lower and upper bounds in zone-based abstractions of timed automata.
\newblock {\em STTT}, 8(3):204--215, 2006.

\bibitem{Boigelot:CAV:2006}
B.~Boigelot and F.~Herbreteau.
\newblock The power of hybrid acceleration.
\newblock In {\em CAV}, volume 4144 of {\em LNCS}, pages 438--451, 2006.

\bibitem{BozIosKon12}
M.~Bozga, R.~Iosif, and F.~Kone{\v c}n{\'y}.
\newblock Deciding conditional termination.
\newblock In {\em TACAS}, volume 7214 of {\em LNCS}, pages 252--266, 2012.

\bibitem{Bozga:ICALP:2006}
M.~Bozga, R.Iosif, and Y.~Lakhnech.
\newblock Flat parametric counter automata.
\newblock In {\em ICALP}, volume 4052 of {\em LNCS}, pages 577--588, 2006.

\bibitem{Cerna:MFCS:2003}
I.~Cern{\'{a}} and R.~Pel{\'{a}}nek.
\newblock Relating hierarchy of temporal properties to model checking.
\newblock In {\em MFCS}, volume 2747 of {\em LNCS}, pages 318--327, 2003.

\bibitem{Comon:CAV:1998}
H.~Comon and Y.~Jurski.
\newblock Multiple counters automata, safety analysis and presburger
  arithmetic.
\newblock In {\em CAV}, volume 1427 of {\em LNCS}, pages 268--279, 1998.

\bibitem{ComJur99}
H.~Comon and Y.~Jurski.
\newblock Timed automata and the theory of real numbers.
\newblock In {\em {CONCUR}}, volume 1664 of {\em LNCS}, pages 242--257, 1999.

\bibitem{CooGulLev08}
B.~Cook, S.~Gulwani, T.~Lev{-}Ami, A.~Rybalchenko, and M.~Sagiv.
\newblock Proving conditional termination.
\newblock In {\em CAV}, volume 5123 of {\em LNCS}, pages 328--340, 2008.

\bibitem{Couvreur:FM:1999}
J.{-}M. Couvreur.
\newblock On-the-fly verification of linear temporal logic.
\newblock In {\em FM}, volume 1708 of {\em LNCS}, pages 253--271, 1999.

\bibitem{DT:TACAS:1998}
C.~Daws and S.~Tripakis.
\newblock Model checking of real-time reachability properties using
  abstractions.
\newblock In {\em TACAS}, volume 1384 of {\em LNCS}, pages 313--329, 1998.

\bibitem{Dill:AVMFSS:1989}
D.~L. Dill.
\newblock Timing assumptions and verification of finite-state concurrent
  systems.
\newblock In {\em Automatic Verification Methods for Finite State Systems},
  pages 197--212, 1989.

\bibitem{GS:MEMICS:2009}
A.~Gaiser and S.~Schwoon.
\newblock Comparison of algorithms for checking emptiness on {B}{\"u}chi
  automata.
\newblock In {\em MEMICS}, volume~13 of {\em OASICS}, pages 69--77, 2009.

\bibitem{GupHenMaj08}
A.~Gupta, T.~A. Henzinger, R.~Majumdar, A.~Rybalchenko, and Ru{-}Gang Xu.
\newblock Proving non-termination.
\newblock In {\em POPL}, pages 147--158. {ACM}, 2008.

\bibitem{NonConvex-ArXiv}
F.~Herbreteau, D.~Kini, B.~Srivathsan, and I.~Walukiewicz.
\newblock Using non-convex approximations for efficient analysis of timed
  automata.
\newblock {\em CoRR}, abs/1110.3704, 2011.

\bibitem{Herbreteau:ATVA:2010}
F.~Herbreteau and B~Srivathsan.
\newblock Efficient on-the-fly emptiness check for timed b{\"u}chi automata.
\newblock In {\em ATVA}, pages 218--232. Springer, 2010.

\bibitem{Herbreteau:LICS:2012}
F.~Herbreteau, B~Srivathsan, and I.~Walukiewicz.
\newblock Better abstractions for timed automata.
\newblock In {\em LICS}, pages 375--384. IEEE Computer Society, 2012.

\bibitem{Herbreteau:CAV:2013}
F.~Herbreteau, B.~Srivathsan, and I.~Walukiewicz.
\newblock Lazy abstractions for timed automata.
\newblock In {\em CAV}, volume 8044 of {\em LNCS}, pages 990--1005, 2013.

\bibitem{Holzmann:DIMACS:1997}
G.~J. Holzmann, D.~Peled, and M.~Yannakakis.
\newblock On nested depth first search.
\newblock {\em DIMACS Series in Discrete Mathematics and Theoretical Computer
  Science}, 32:23--31, 1997.

\bibitem{Laarman:CAV:2013}
A.~Laarman, Olesen M.~C., Dalsgaard A.~E., Larsen K.~G., and J.~van~de Pol.
\newblock Multi-core emptiness checking of timed b{\"u}chi automata using
  inclusion abstraction.
\newblock In {\em CAV}, volume 8044 of {\em LNCS}, pages 968--983, 2013.

\bibitem{Li:FORMATS:2009}
G.~Li.
\newblock Checking timed {B}{\"u}chi automata emptiness using
  {LU}-abstractions.
\newblock In {\em FORMATS}, pages 228--242, 2009.

\bibitem{Tiw04}
A.~Tiwari.
\newblock Termination of linear programs.
\newblock In {\em CAV}, volume 3114 of {\em LNCS}, pages 70--82, 2004.

\bibitem{Tripakis:FMSD:2001}
S.~Tripakis and S.~Yovine.
\newblock Analysis of timed systems using time-abstracting bisimulations.
\newblock {\em Formal Methods in System Design}, 18(1):25--68, 2001.

\bibitem{UPPAAL-CSMACD}
{UPPAAL} {CSMA}/{CD} model.
\newblock
  \url{http://www.it.uu.se/research/group/darts/uppaal/benchmarks/genCSMA_CD.awk}.
\newblock Accessed: 2014-10-08.

\end{thebibliography}


\begin{thebibliography}{1}
\bibitem{Yi94}
W. Yi and P. Pettersson and M. Daniels.
\newblock{\em Automatic verification of real-time communicating
  systems by constraint-solving}
\newblock{In {\em {FORTE}}, volume 6 of {\em {IFIP} Conference
    Proceedings}, pages 243--258, 1994.}

\bibitem{TrainGateUPPAAL}
{UPPAAL} Train Gate model.
\newblock{\em Provided with UPPAAL model-checker:
  \url{http://www.uppaal.org}.}
\newblock{Accessed: 2014-10-08}

\bibitem{Tripakis01}
S. Tripakis and S. Yovine.
\newblock{Analysis of Timed Systems Using Time-Abstracting
  Bisimulations.}
\newblock{{\em Formal Methods in System Design}, 18(1):25--68, 2001}

\bibitem{FischerUPPAAL}
{UPPAAL} {Fischer's Mutex} model.
\newblock{\url{http://www.it.uu.se/research/group/darts/uppaal/benchmarks/genFischer.awk},}
\newblock{Accessed: 2014-10-08}

\bibitem{CsmaCdUPPAAL}
{UPPAAL} {CSMA}/{CD} model.
\newblock{\url{http://www.it.uu.se/research/group/darts/uppaal/benchmarks/genCSMA_CD.awk},}
\newblock{Accessed: 2014-10-08}

\bibitem{Daws96}
C. Daws and A. Olivero and S. Tripakis and S. Yovine.
\newblock{The tool {KRONOS},}
\newblock{in {\em Hybrid Systems III}, volume 1066 of {\em LNCS},
  pages 208--219, 1996.}

\bibitem{FddiUPPAAL}
{UPPAAL} {Token ring FDDI protocol} model.
\newblock{\url{http://www.it.uu.se/research/group/darts/uppaal/benchmarks/genHDDI.awk},}
\newblock{Accessed: 2014-10-08}

\bibitem{Herbreteau13}
F. Herbreteau and B. Srivathsan and I. Walukiewicz.
\newblock{Lazy Abstractions for Timed Automata.}
\newblock{In {\em {CAV}}, volume 8044 of {\em LNCS}, pages 990--1005,
  2013.}

\bibitem{Laarman13}
A. Laarman and M. C., Olesen and A. E., Dalsgaard and K. G., Larsen
and J. van de Pol.
\newblock{Multi-core emptiness checking of timed B{\"u}chi automata
  using inclusion abstraction}
\newblock{In {\em {CAV}}, volume 8044 of {\em LNCS}, pages 968--983,
  2013.}
\end{thebibliography}

\newpage
\appendix

\section{DFSS and iDFSS algorithms}
\label{sec:appendix:algo}

\begin{lstlisting}[
caption={iDFSS algorithm.},
label={lst:iDFSS},
float=tbp]
procedure iDFSS($\Aa = (Q, q_0, X, T, F)$)
  Cyan:=Blue:=$\es$;
  explore($(q_0, Z_0)$);
  report no cycle;

procedure explore($(q,Z)$)
  Cyan:=Cyan$\, \cup \, \set{(q,Z)}$;
  for all $(q',Z')\in \mathsf{POST}((q,Z))$ do
    if $q' \not \in F$ and $(q',Z') \in Cyan$ then £\label{lst:iDFSS:cyan-membership}£
      skip //ignore $(q',Z')$
    if $q'\in F$ then
      if $\exists (q',Z'')\in Cyan.\ Z''\incl Z'$ then  report cycle; £\label{lst:iDFSS:cyan-inclusion}£
      if $\exists (q',Z'')\in Cyan$ then £\label{lst:iDFSS:iterability-check}£
         let $\s$ be the path $(q',Z'')\to (q',Z')$ on the stack;
         if $\w$-iterable($\s$) then report cycle;
    if $\exists (q',Z'')\in Blue.\ Z'\incl Z''$ then £\label{lst:iDFSS:blue-inclusion}£
      skip //ignore $(q',Z')$
    else
      explore($(q',Z')$);
  Blue:=Blue$\, \cup \, \set{(q,Z)}$;
  Cyan:=Cyan$\, \setminus \, \set{(q,Z)}$;
\end{lstlisting}


Algorithm iDFSS is presented in Listing~\ref{lst:iDFSS}. This
algorithm runs a depth-first search over the finite abstracted zone
graph $\ZG^{\extraLUp}(\Aa)$ of a timed Büchi automaton $\Aa$ to check
whether $\Aa$ has an accepting run. To that purpose, the algorithm
maintains two sets of states. $Blue$ contains fully visited states and
$Cyan$ consists in partially visited states that form the current
search path.

Algorithm iDFSS uses the information provided by the zones to limit
the seach from the current node $(q',Z')$ in three cases:
\begin{enumerate}
\item if $q'$ is accepting and on the stack there is a path $\sigma$
  from a node $(q',Z'')$ to the current node with $Z'' \incl Z'$ then
  report existence of an accepting path
  (line~\ref{lst:iDFSS:cyan-inclusion}).
\item when $q'$ is accepting but $Z'' \not \incl Z'$ we use
  $\w$-iterability check on $\sigma$. If $\sigma$ is iterable from
  $(q',Z')$, we have detected an accepting path and we stop the search
  (line~\ref{lst:iDFSS:iterability-check}).
\item if there is a fully explored a node $(q',Z'')$, i.e a node not
  on the stack, with $Z' \incl Z''$ then ignore the current node and
  return from the DFS call (line~\ref{lst:iDFSS:blue-inclusion}).
\end{enumerate}

In this paper, we compare algorithm iDFSS to algorithm DFSS that does
not use iterability check (i.e. with
line~\ref{lst:iDFSS:iterability-check} and the next two lines
removed). This is the algorithm one obtains after specialisation of
the algorithm of Laarman \textit{et al.}~\cite{Laarman:CAV:2013} to
weak Buchi properties.


\section{Models and properties}\label{sec:models}

In our experiments we have used the following models:
\begin{itemize}
\item Train gate controller \cite{Yi94,TrainGateUPPAAL}
\item Fischer's Mutex protocol
  \cite{Tripakis01,FischerUPPAAL}
\item CSMA/CD \cite{Tripakis01,CsmaCdUPPAAL}
\item FDDI \cite{Daws96,FddiUPPAAL}
\end{itemize}

The above models have been used as benchmarks to evaluate
algorithms for safety and
liveness~\cite{Herbreteau13,Laarman13}.  
As we have explained in the main text the standard benchmark CSMA/CD
model has a transition missing. In Figure~\ref{fig:csma} we
graphically represent the model with 3 stations, the new behaviours
due to adding the missing transition are marked in red. The intention
of this figure is just to give a visual impression of what is the
influence of the added transition. After adding the transition there
are no new states, but there are new behaviours.

\begin{figure}[tbhp]
  \centering
  \includegraphics[width=18cm,angle=90,origin=c]{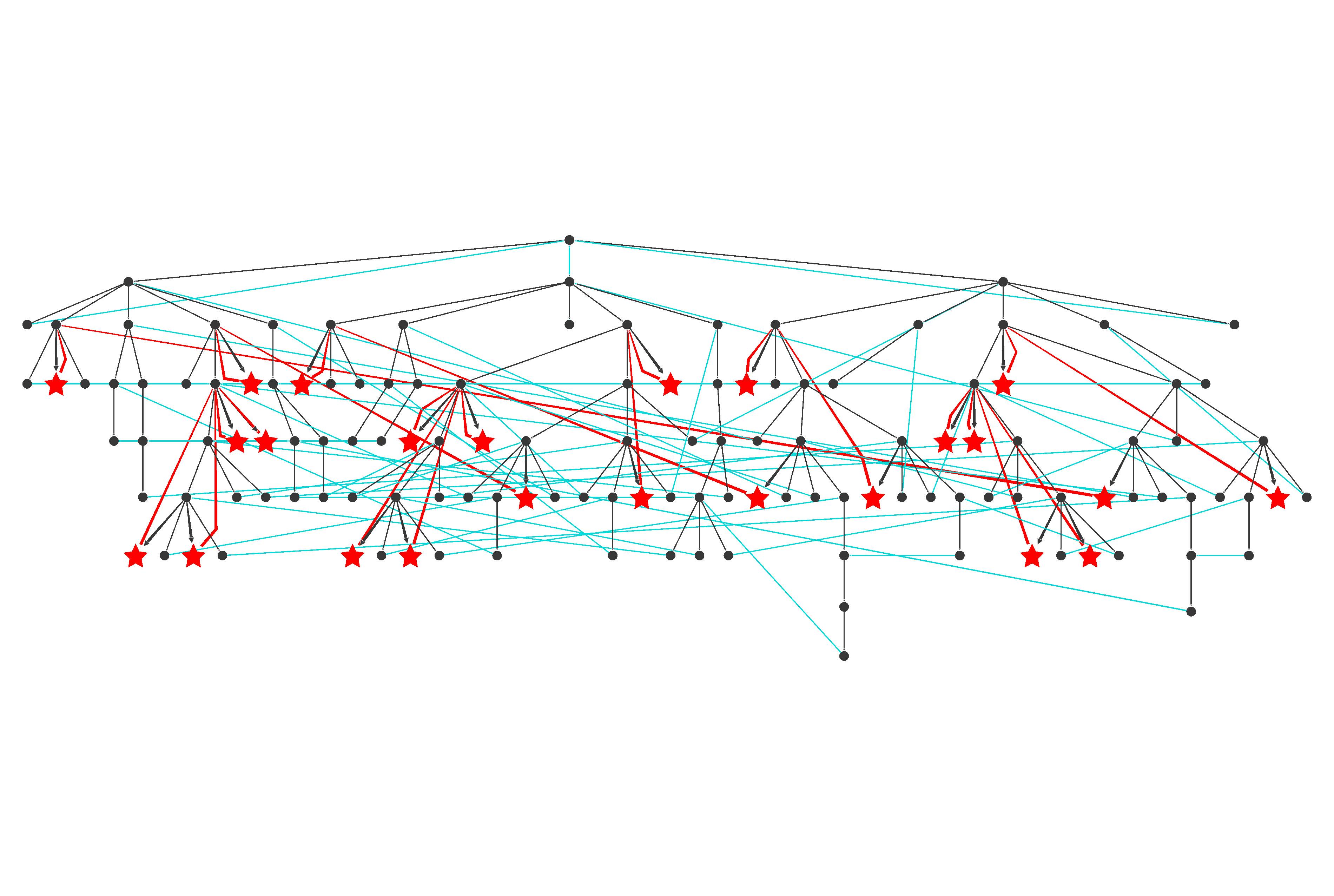}
  \caption{CSMA/CD model: after correcting the model we obtain new
    cycles marked in red.}
  \label{fig:csma}
\end{figure}

\begin{figure}
  \centering
  \begin{tikzpicture}[>=stealth',shorten >=1pt,auto,node distance=2.5cm]
    \tikzstyle{every node}=[font=\tiny]
    \node[initial,state,accepting] (q0)               {$o_1\leq 15T$};
    \node[state]                   (q1) [right of=q0] {};

  \path[->] (q0) edge [loop above,looseness=5] node[align=center] 
                      {$10T\leq o_2<15T$,$enter_1$,$\{o_2\}$}        (q0)
            (q0) edge [loop below,looseness=5] node[align=center] 
                      {$o_1\leq T$,$req_1$,$\{o_1\}$}                (q0)
                 edge node[auto,text width=2cm, align=center] 
                      {$o_2\geq 15T$\\$enter_1$}                     (q1);
\end{tikzpicture}
\caption{Property for model Fischer with $N$ processes and
  $T=K*N$. The automaton checks if the first process can request fast,
  and can only enter the critical section after a certain amount of
  time and before timeout.}
\end{figure}
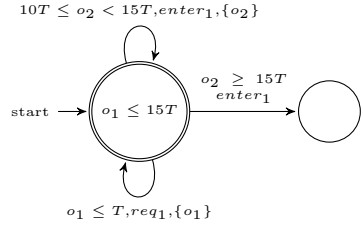

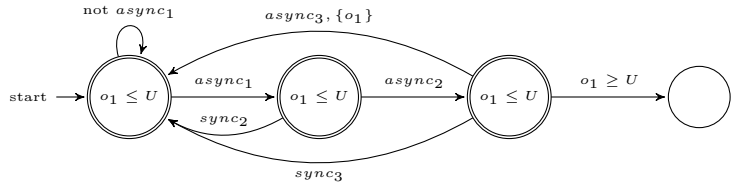
\begin{figure}
  \centering
  \begin{tikzpicture}[>=stealth',shorten >=1pt,auto,node distance=2.5cm]
    \tikzstyle{every node}=[font=\tiny]
    \node[initial,state,accepting] (q0)                 {$o_1 \leq U$};
    \node[state,accepting]              (q1)   [right of=q0] {$o_1\leq U$};
    \node[state,accepting]              (q2)   [right of=q1] {$o_1\leq U$};
    \node[state]                        (q3) [right of=q2] {};

  \path[->] (q0) edge [loop above,looseness=5] node{not $async_1$}        (q0)
            (q0) edge node{$async_1$}                         (q1)
            (q1) edge node{$async_2$}                         (q2)
            (q2) edge [bend right] node[above]{$async_3, \{o_1\}$} (q0)

            (q1) edge [bend left] node[above]{$sync_2$} (q0)
            (q2) edge [bend left] node{$sync_3$}        (q0)

            (q2) edge node{$o_1 \geq U$}                (q3)
;

\end{tikzpicture}
\caption{Property for model FDDI with $N=3$ stations and $U =
  150*SA*N$. The automaton checks if all stations can send
  asynchronous messages in the same round infinitely often and in a
  bounded amount of time.}
\end{figure}

\begin{figure}
  \centering
  \begin{tikzpicture}[>=stealth',shorten >=1pt,auto,node distance=3cm, align=center]
    \tikzstyle{every node}=[font=\tiny]
    \node[initial,state,accepting] 
                            (q0)                {$o_2\leq30N$};
    \node[state,accepting, text width=1cm]
                            (q1) [right of= q0] {};
    \node[state,accepting]  (q2) [below of=q1]  {$o_1\leq300N$};
    \node[state]            (q3) [right of=q1]  {};

  \path[->] (q0) edge              node[text width=3cm] 
                                   {$o_2\leq30N$\\$approach1$\\$\{o_2\}$} (q1)
            (q1) edge [bend left]  node {$leave_1$}                       (q0)
                 edge              node[text width=3cm, align =left] 
                                   {$o_1\geq 300N$\\$stop_1$\\$\{o_1\}$}   (q2)
                 edge              node[auto, text width=2cm] 
                                   {$o1<300N$\\$stop_1$}                  (q3)
            (q2) edge              node {$leave_1$}                       (q0);
\end{tikzpicture}
\caption{Property for model Train Gate, with $N$ trains. The automaton
  checks if Train $1$ can approach frequently fast but hardly waits
  for other trains.}
\end{figure}
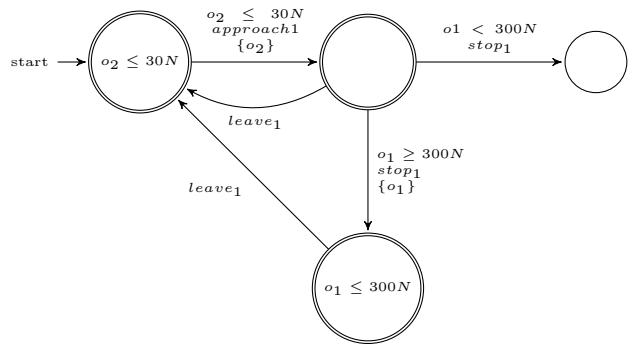


\end{document}